\pdfoutput=1
\documentclass[11pt]{article}

\usepackage[USenglish]{babel}
\usepackage[T1]{fontenc}
\usepackage[utf8]{inputenc}
\usepackage{textcomp}
\usepackage{lmodern}
\usepackage[margin=1in]{geometry}
\usepackage{microtype}
\usepackage{xcolor}
\usepackage{amsmath,amssymb,amsfonts,mathtools}
\usepackage{graphicx}
\usepackage{array,multirow,tabularx,threeparttable,longtable}
\usepackage{enumerate}
\usepackage{comment}
\usepackage{tikz}
\usepackage[linesnumbered,ruled,vlined]{algorithm2e}
\usepackage{float}
\usepackage{caption}
\usepackage{subcaption}
\usepackage{rotating}
\usepackage{textpos}
\usepackage{amsthm}
\usepackage{thmtools}
\usepackage{thm-restate}
\usepackage[textsize=footnotesize,colorinlistoftodos]{todonotes}
\usepackage{hyperref}
\usepackage{cleveref}
\usetikzlibrary{calc}

\numberwithin{equation}{section}

\theoremstyle{plain}
\newtheorem{theorem}{Theorem}[section]
\newtheorem{lemma}[theorem]{Lemma}
\newtheorem{proposition}[theorem]{Proposition}

\newtheorem{fact}[theorem]{Fact}
\theoremstyle{definition}
\newtheorem{definition}[theorem]{Definition}
\theoremstyle{remark}
\newtheorem{remark}[theorem]{Remark}
\newtheorem*{theorem*}{Theorem}
\newtheorem*{corollary*}{Corollary}

\crefname{observation}{Observation}{Observations}
\Crefname{observation}{Observation}{Observations}
\crefname{fact}{Fact}{Facts}
\Crefname{fact}{Fact}{Facts}
\crefname{algorithm}{Algorithm}{Algorithms}
\Crefname{algorithm}{Algorithm}{Algorithms}

\hypersetup{
  colorlinks=true,
  linkcolor=blue!75!black,
  citecolor=green!75!black,
  urlcolor=blue!75!black,
}

\definecolor{MidnightBlack}{rgb}{0.1,0.1,.34}
\definecolor{MidnightBlue}{rgb}{0.1,0.1,0.43}
\definecolor{Black}{rgb}{0,0, 0}
\definecolor{Blue}{rgb}{0, 0 ,1}
\definecolor{Red}{rgb}{1, 0 ,0}
\definecolor{White}{rgb}{1, 1, 1}
\definecolor{grey}{rgb}{.6, .6, .6}
\definecolor{Mygreen}{rgb}{.0, .7, .0}
\definecolor{Yellow}{rgb}{.55,.55,0}
\definecolor{Mustard}{rgb}{1.0, 0.86, 0.35}
\definecolor{applegreen}{rgb}{0.55, 0.71, 0.0}
\definecolor{darkturquoise}{rgb}{0.0, 0.81, 0.82}
\definecolor{celestialblue}{rgb}{0.29, 0.59, 0.82}
\definecolor{green_yellow}{rgb}{0.68, 1.0, 0.18}
\definecolor{crimsonglory}{rgb}{0.75, 0.0, 0.2}
\definecolor{darkmagenta}{rgb}{0.30, 0.0, 0.30}
\definecolor{magenta}{rgb}{0.50, 0.0, 0.50}
\definecolor{internationalorange}{rgb}{1.0, 0.31, 0.0}
\definecolor{darkorange}{rgb}{1.0, 0.55, 0.0}
\definecolor{ao}{rgb}{0.0, 0.5, 0.0}
\definecolor{awesome}{rgb}{1.0, 0.13, 0.32}
\definecolor{darkcyan}{rgb}{0.0, 0.50, 0.50}
\definecolor{violet}{rgb}{0.93, 0.51, 0.93}
\definecolor{brown}{rgb}{0.65, 0.16, 0.16}
\definecolor{orange}{rgb}{1.0, 0.65, 0.0}
\definecolor{cornflowerblue}{rgb}{0.39, 0.58, 0.93}

\newcommand{\blue}[1]{{\color{Blue}#1}}
\newcommand{\red}[1]{{\color{Red}#1}}

\usetikzlibrary{arrows.meta,calc,positioning,backgrounds}

\definecolor{navy}{RGB}{33,42,110}
\definecolor{lightblue}{RGB}{198,214,241}
\definecolor{lightgreen}{RGB}{200,238,210}
\definecolor{lightyellow}{RGB}{243,221,123}
\definecolor{edgegreen}{RGB}{16,150,102}
\definecolor{edgeblue}{RGB}{66,111,204}
\definecolor{edgered}{RGB}{238,70,70}
\definecolor{textgray}{RGB}{140,156,183}
\definecolor{arrowgray}{RGB}{95,112,140}
\definecolor{bggray}{RGB}{246,246,246}

\makeatletter
\newcommand{\linkedstatement}[3]{%
  \begingroup
    \let\thmt@origtrivialref\thmt@trivialref
    \def\thmt@trivialref##1##2{%
      \ifnum\pdfstrcmp{##1}{thmt@@#1}=0\relax
        \hyperref[#2]{\thmt@origtrivialref{##1}{##2}}%
      \else
        \thmt@origtrivialref{##1}{##2}%
      \fi
    }%
    #3%
  \endgroup
}
\makeatother
\tikzset{
  basevertex/.style={circle, draw=navy, line width=1.6pt, minimum size=11mm, inner sep=0pt, font=\bfseries\Large, text=navy},
  leftvertex/.style={basevertex, fill=lightblue},
  vzero/.style={basevertex, fill=lightblue},
  vone/.style={basevertex, fill=lightgreen},
  auxvertex/.style={basevertex, fill=lightyellow},
  darkedge/.style={draw=navy, line width=1.8pt},
  gedge/.style={draw=edgegreen, line width=1.7pt},
  bedge/.style={draw=edgeblue, line width=1.5pt},
  title/.style={font=\rmfamily\bfseries\fontsize{23}{25}\selectfont, text=navy},
  sideanno/.style={font=\itshape\small, text=black, align=left},
  bodytext/.style={font=\large},
  legendtext/.style={font=\large, align=left}
}

\title{W-state graphs: Structure and Algorithms}

\author{
Rishikesh Gajjala\thanks{Supported by Center for Quantum and Topological Systems, NYUAD.}\\
New York University Abu Dhabi\\
\texttt{r.gajjala@nyu.edu}
\and
Saurabh Ray$^*$\\
New York University Abu Dhabi\\
\texttt{saurabh.ray@nyu.edu}
\and
Dimitrios M. Thilikos\thanks{Supported by the Franco-Norwegian project PHC AURORA 2024-25 (Projet n°\! 51260WL) and the French National Research Agency (ANR) under project GODASse ANR-24-CE48-4377 and under the France 2030 grant reference number ANR-24-RRII-0002 operated by the Inria Quadrant Program.}\\
LIRMM, Univ Montpellier, CNRS, Montpellier, France\\
\texttt{sedthilk@thilikos.info}
}

\date{}
\newcommand{\paperkeywords}{W-state graphs, perfect matchings, matching-covered graphs, factor-critical graphs, structural graph theory}

\begin{document}
\maketitle

\begin{abstract}
\noindent We study the class of edge-coloured graphs arising from the graph-theoretic representation of quantum photonic experiments that generate multipartite W-states. Abstracting away physical amplitudes and phases, we introduce \emph{W-state graphs}: matching-covered graphs equipped with a half-edge $2$-colouring such that every perfect matching contains exactly one bichromatic edge and every vertex is incident with a \red{red} half-edge. Our main contribution is a complete structural characterization of W-state graphs. We show that a graph is a W-state graph if and only if each of its $3$-connected components is a W-cone, a simple and rigid building block defined by a universal vertex and a factor-critical base. This characterization implies that no W-state graph is simple and yields a recognition algorithm running as fast as verifying whether a graph is matching-covered. We also show that the natural generalization to Dicke states encounters a complexity barrier: verifying one of the two Dicke state conditions is itself coNP-complete, resolving an open problem of Vardi and Zhang [IJCAI 2023]. Our results place W-state graphs firmly within classical matching theory and precisely delineate the combinatorial structures capable of realizing idealized W-states in the experiment-graph framework.
\end{abstract}

\medskip

\newpage 

\section{Introduction}
Motivated by the graph-theoretic representation of photonic quantum experiments, we study the  graphs arising  when the coincidence structure of an experiment is required to realize an idealized $n$-partite W-state. Abstracting away all physical amplitudes and phases, we isolate a purely combinatorial notion—W-state graphs—defined by local colouring constraints and a global restriction on perfect matchings. The central question we address is structural: which graphs satisfy these constraints, and how rigid are they? Our main result gives a complete characterization of W-state graphs in terms of classical matching-theoretic structure, yielding both conceptual clarity and algorithmic consequences. This characterization  delineates the exact combinatorial boundary of W-state realizability within the experiment-graph framework~\cite{Quantum_graphs, Quantum_graphs_3}.

\subsection{Motivation}

Multipartite entangled states such as GHZ and W-states form two of the most prominent and inequivalent classes of quantum entanglement. While GHZ states display extreme non‑locality, they are very fragile under particle loss; in contrast, W-states retain bipartite entanglement even if any one subsystem is traced out, which makes them attractive resources for quantum communication, distributed sensing, and other multi‑party protocols~\cite{DurVidalCirac00,wstatefml}. This has led to a sustained effort, especially in photonic platforms, to design experiments that generate high‑quality W-states in a scalable and resource‑efficient way.

Krenn, Gu and Zeilinger introduced a method to represent a large class of quantum photonic experiments as edge‑weighted, half‑edge–coloured multigraphs \cite{Quantum_graphs}. In this representation, the vertices correspond to optical modes, the edges to photon‑pair sources, and perfect matchings encode coincidence detection events; the quantum state produced by the experiment is then a coherent superposition of basis states indexed by perfect matchings. The combinatorial structure of perfect matchings and edge colours largely determines which multipartite entangled state is realized in the laboratory. Additionally, every edge-coloured edge-weighted graph can be translated into a concrete experimental setup.

This technique has led to the discovery of new quantum interference effects and connections to quantum computing \cite{Quantum_graphs_3}. Furthermore, it has been used as the representation of efficient AI-based design methods for new quantum experiments \cite{AIquantum1, AIquantum2}. The states formed by this framework were also experimentally demonstrated \cite{Feng:23, qian2023multiphoton}. This graph-based representation was also used to demonstrate many more systems beyond post-selected states (like NOON states and heralded states) \cite{AIquantum2}. This representation and another closely related graph-based representation have also been used for quantum circuit representation and computation \cite{AIquantum1, anand2022information}. The largest integrated photonic chip experiment (with several applications) presented so far \cite{Bao23} also follows a graph-based representation. 

On the other hand, most of the graph‑theoretic work so far has focused on the Greenberger–Horne–Zeilinger (GHZ) states, where the amplitudes of all \emph{monochromatic vertex colourings} are equal, and all other colourings cancel. For this case, several structural, extremal and algorithmic questions have been answered: which graphs admit such a colouring \cite{CG1}, how large the “dimension’’ (number of distinct monochromatic colourings) can be \cite{Quantum_graphs, CG2}, and can one efficiently verify these properties \cite{vardi1, vardi2}. In contrast, the analogous question for W-states is far less understood: which graphs encode experiments whose perfect matchings coherently realize an $n$‑partite W-state?

In this paper, we take a purely graph‑theoretic approach to this question. We isolate the combinatorial core of the W‑state condition and arrive at the notion of a \emph{W‑state graph}, a matching‑covered half‑edge 2‑coloured multigraph in which every perfect matching contains \emph{exactly one bichromatic edge} and every vertex is incident with at least one \emph{\red{red} half‑edge} in a perfect matching (\Cref{def_Wstate}). This definition abstracts away amplitudes, phases, and other physical details while retaining exactly those constraints that are forced by the target W-state. Our main goal is to understand the structure of such graphs in classical graph-theoretic terms and thereby obtain a clean characterization that can be used in the design of quantum photonic experiments.

\subsection{Related work}


 The combinatorial constraint we study is closely related to the classical Exact Matching (EM) problem of Papadimitriou and Yannakakis~\cite{DBLP:conf/icalp/PapadimitriouY79}. In the EM problem one is given an edge-coloured graph with colours \red{red} and \blue{blue}, and an integer $k$, and the goal is to decide whether the graph contains a perfect matching with exactly $k$ \red{red} edges. In their seminal work, Mulmuley, Vazirani and Vazirani~\cite{DBLP:journals/combinatorica/MulmuleyVV87} showed that EM is in $\mathsf{RP}$ (indeed, in $\mathsf{RNC}$) via the isolation lemma. It is one of the canonical natural combinatorial problems known to be in $\mathsf{RP}$ but not known to lie in $\mathsf{P}$.  Despite several efforts, the progress has been limited to a few restricted graph classes \cite{DBLP:conf/stacs/Maalouly23,DBLP:journals/algorithmica/Yuster12}. We refer the reader to \cite{DBLP:phd/basesearch/Maalouly24} to get a more detailed overview of this area. The W-state graphs (see \Cref{def_Wstate}) satisfy a universal (“for all”) variation of EM with parameter $k=1$: \emph{every} perfect matching (rather than at least one perfect matching) contains exactly one \emph{bichromatic edge} (analogous to a \red{red} edge).

Vardi and Zhang \cite{vardi1, vardi2} developed a complexity‑theoretic perspective formalising \emph{perfect matching under vertex‑colour constraints} (PM‑VC) as graph problems in which the legal vertex colourings encode target quantum states (such as GHZ, W, or Dicke states) and studied their algorithmic complexity on bicoloured graphs. Chandran and Gajjala initiated a systematic study of graphs that encode GHZ-states \cite{CG1}. They gave a structural characterization of edge-coloured graphs that have only monochromatic perfect matchings. This yields, in particular, a complete description of \textit{GHZ-state graphs}. For the more general edge-weighted version of GHZ-state graphs, the structure \cite{CGI3} and bounds on the dimension were also studied \cite{CerveraLierta2022designofquantum, CG2}. 



To the best of our knowledge, there is no prior structural characterization of graphs that encode W-states in the experiment‑graph sense. Only a small subclass of W-state graphs has been found so far \cite{Quantum_graphs_3}, and several advanced automated
search methods are being used to find more \cite{AIquantum1, AIquantum2}. Our work fills this gap by identifying the right purely combinatorial notion (W‑state graphs) and describing exactly which bicoloured graphs satisfy it.

\subsection{Quantum photonic experiments as graphs}\label{quanttograph}

We briefly recall the correspondence between quantum photonic experiments and half‑edge-coloured graphs, specializing to the case relevant for W-states. A typical setup consists of several nonlinear crystals acting as photon‑pair sources, linear‑optical elements (beam splitters, phase shifters, interferometers), and a collection of single‑photon detectors. Each crystal, when pumped, probabilistically emits a pair of photons into two optical paths; these photons may carry a discrete “mode number’’ (for example, an orbital‑angular‑momentum mode). A \emph{coincidence} event is a run of the experiment in which each detector registers exactly one photon.

The experiment‑graph formalism associates to such a setup a half‑edge-coloured, edge-weighted multigraph $H$. Vertices of $H$ correspond to optical modes or output ports, and each crystal is represented by an edge whose endpoints are the two modes to which the crystal can send photons. If a crystal can emit a photon in different modes on either arm, these modes are recorded as colours on the corresponding half‑edges. Additionally, each edge carries a complex weight encoding the amplitude of the photon pairs. Under mild assumptions on the experimental setup, there is then a one‑to‑one correspondence between coincidences and perfect matchings of $H$: activating exactly those crystals whose edges form a perfect matching results in one photon per detector.

Given a perfect matching $P$ of $H$, one obtains an \emph{inherited vertex colouring} by assigning to each vertex the colour of the half‑edge incident with it in $P$ and an amplitude equal to the product of the edge weights along $P$. The overall quantum state generated by the experiment is a superposition over all perfect matchings $P$, where each basis vector is labelled by an inherited vertex colouring and its coefficient is the sum of amplitudes of all perfect matchings that induce this colouring. Target states such as GHZ, W, and Dicke states can therefore be expressed as conditions on which vertex colourings appear with non‑zero amplitude and how their amplitudes relate to each other.
The description of an $n$‑partite W-state is given by \[
\lvert W_n \rangle
= \frac{1}{\sqrt{n}}
\left(
\lvert 10\ldots 0 \rangle
+ \lvert 010\ldots 0 \rangle
+ \cdots
+ \lvert 0\ldots 01 \rangle
\right).
\]
The computational basis states that appear are precisely those in which exactly one party is in the excited state and the remaining $n-1$ parties are in the ground state and all these basis states have equal amplitude up to a global phase. In the experiment‑graph picture, and in the absence of destructive interference, this suggests a combinatorial constraint: every perfect matching should induce a vertex colouring with exactly one vertex in colour~$1$ (say \red{red}) and $n-1$ vertices in colour~$0$ (say \blue{blue}), and for each vertex there should exist some perfect matching in which that vertex carries colour~$1$. When we identify the excited state with the colour carried by a \red{red} half‑edge, these constraints translate directly into the definition of a W‑state graph (\Cref{def_Wstate}): every perfect matching uses exactly one bichromatic edge (the location of the “excited’’ vertex) and every vertex is incident with a \red{red} half‑edge on at least one edge. It is easy to see that the crystals (edges) which are not part of a coincidence (perfect matching) are redundant. Also the crystals (edges) in which both the photons are excited (monochromatic \red{red} edges) are redundant. Therefore, we restrict our attention to edge-coloured graphs which are matching-covered and have no monochromatic \red{red} edges.   

The rest of the paper develops the consequences of this abstraction. We forget about the quantum mechanical terms and work purely with matching‑covered half‑edge 2‑coloured graphs satisfying the matching and vertex conditions of \Cref{def_Wstate}. We show that this class of graphs admits a clean structural description in terms of classical matching theory. {Given an uncoloured graph $G$, our structural result leads to an algorithm for deciding whether there exists a half-edge
$2$-colouring $c$ such that $(G=(V,E),c)$ is a W-state graph,  
in $O(|V|\cdot |E|)$ time.}
From the physical perspective, this characterization describes precisely which connectivity patterns of photon‑pair sources and modes can realize idealized W-states in the experiments.

\section{Graph-theoretic formulation and our result}
\subsection{Notation and Preliminaries}
Graphs considered in this paper may have multiple edges but no loops. 
Let \(G=(V,E)\) be a finite undirected graph where $V$ is the set of its vertices, denoted by $V(G)$, and $E$ is the set of its edges, denoted by $E(G)$. We refer to $|V|$ as the \emph{order} of $G$. We call a vertex in G that is adjacent to all other vertices of G a \emph{universal} vertex. We denote the set of all perfect matchings of $G$ as $\mathcal{PM}(G)$. 
Given a vertex set $S\subseteq V(G)$ we denote by $G[S]$ the subgraph of $G$ 
induced by $S$.  We also refer to $V(G)\setminus S$ as the \emph{complement}
of $S$ in $V(G)$. For two disjoint vertex sets $X, X' \subseteq V$, the \emph{cut} between them is denoted as $\delta(X, X') \coloneqq  \{\, \{u, v\} \in E : u \in X,\; v \in X' \,\}$. When $X'=V(G)\setminus X$, we simplify the notation of $\delta(X, X')$ as $\delta(X)$. If $G$ has odd order $n$, a matching of $G$ containing $\frac{n-1}{2}$ edges is called a \emph{near-perfect} matching of $G$.

\begin{definition}[Matching-covered graph]\label{def_MC}
A graph \(G\) is \emph{matching-covered} if and only if \emph{every} edge of \(G\) lies in at least one perfect matching of \(G\).
\end{definition}


\begin{definition}[Half-edge $2$-colouring]\label{def_half2}
A \emph{half–edge $2$–colouring} of \(G\) is a function
\[
  c:\Bigl\{(e,w)\mid e=\{u,v\}\in E,\;w\in\{u,v\}\Bigr\}\longrightarrow\{0,1\}.
\]
\end{definition}
In this paper, we refer to the colour $0$ as \blue{blue} and the colour $1$ as \red{red}. For $e=\{u,v\}$, one can visualize the half-edge $2$-colouring \(c(e,v)\) as the colour assigned to the half of edge \(e\) incident to \(v\).

\begin{definition}[Monochromatic edge]
We call an edge $e=\{u,v\}$ \emph{monochromatic} if \(c(e,u)=c(e,v)\)  and \emph{bichromatic} otherwise. 
\end{definition}
For a monochromatic edge \(e=\{u,v\}\), we define \(c(e)=c(e,u)=c(e,v)\). We can partition the edge set $E(G)$ as \(E=E_m(G)\sqcup E_b(G)\), where \(E_m(G)\) contains all monochromatic edges
and \(E_b(G)\) contains all bichromatic edges.
\begin{remark}
We note that, throughout, unless explicitly stated otherwise, the monochromatic edges are always \blue{blue}, i.e., there are no \red{red} monochromatic edges as discussed in \Cref{quanttograph}. 
\end{remark}

\begin{definition}[Factor-critical or Hypomatchable Graphs \cite{lucchesi2024perfect}]
 A graph \(G=(V,E)\) is called \emph{factor-critical} if and only if, for every vertex \(u\in V\), the graph $G[V\setminus \{u\}]$ admits a perfect matching.   
\end{definition}

\begin{definition}[W-state graph]\label{def_Wstate}
A \emph{W-state graph} is a half-edge $2$-coloured matching-covered graph \((G,c)\) such that
\begin{align}
 &\forall M\in \mathcal{PM}(G),\;\;|M\cap E_b(G)|=1.
  &\!\!\!\text{\emph{(matching condition)}}\label{eq_MC}\\
  &\forall v \in V(G),\;\exists e \in \delta(\{v\}) \;:\; c(e,v)= \red{1}.
  &\!\!\!\text{\emph{(vertex condition)}}\label{eq_VC}
\end{align}
\end{definition}
One can see that the graphs in \Cref{fig_cone} are matching-covered and satisfy both \Cref{eq_MC} and \Cref{eq_VC}. Therefore, they are W-state graphs.

\begin{restatable}{proposition}{WStateTwoConnected}
\label{prop:wstate-two-connected}
W-state graphs are $2$-connected.     
\end{restatable}
\begin{proof}
We first claim that W-state graphs are connected. Suppose not. Towards a contradiction, let there be two or more connected components in a W-state graph $G$. Observe that there exists a perfect matching on each component, containing a bichromatic edge (from the vertex condition and being matching-covered). One could take a disjoint union of these perfect matchings over the components to form a new perfect matching on $G$. Such a perfect matching would contain more than one bichromatic edge, which contradicts the matching condition. So we know that every W-state graph is connected.

It is known that a matching-covered graph cannot have a cut vertex (see \Cref{cut_ver_mcg}). Therefore, W-state graphs are $2$-connected.
\end{proof}

\begin{definition}[W-cone]\label{conedef}
A {\em W-cone} is a half-edge $2$-coloured matching-covered graph $(G,c)$ 
containing a special universal vertex $v$ (referred to as apex vertex) such that 
\begin{align}
 &\delta(\{v\}) = E_b(G)\ \text{ and }\ E\bigl(G[V(G)\setminus \{v\}]\bigr)= E_m(G).
  &\!\!\!\text{\emph{(edge partition condition)}}\label{eq_edge_partition}\\
 &\forall u \in V(G),\;\exists e \in \delta(\{u\}) \;:\; c(e,u)= \red{1}.
  &\!\!\!\text{\emph{(vertex condition)}}\label{eq_vertex_condition}
\end{align}

\end{definition}

\begin{figure}[t]
\centering

\begin{minipage}[t]{0.47\textwidth}
\centering
\begin{tikzpicture}[line cap=round, line join=round]
  \tikzset{
    vert/.style={circle, draw=black, fill=white,
                 minimum size=7mm, inner sep=0pt, font=\footnotesize},
    thickedge/.style={line width=1.1pt},
  }

  \coordinate (v1) at (0,0);
  \coordinate (v2) at (2,0);
  \coordinate (v3) at (5.8,1.6);
  \coordinate (v4) at (5.8,-1.6);

  \coordinate (m12) at ($ (v1)!0.5!(v2) $);
  \coordinate (m14) at ($ (v1)!0.5!(v4) $);

  \coordinate (c13a) at (2.5,1.9);
  \coordinate (c13b) at (2.5,0.9);
  \coordinate (m13a) at (2.7,1.35);
  \coordinate (m13b) at (2.7,0.85);

  \draw[thickedge, blue] (v1) -- (m12);
  \draw[thickedge, red]   (m12) -- (v2);

  \draw[thickedge, blue] (v1) -- (m14);
  \draw[thickedge, red]   (m14) -- (v4);

  \draw[thickedge, blue] (v2) -- (v3);
  \draw[thickedge, blue] (v2) -- (v4);
  \draw[thickedge, blue] (v3) -- (v4);

  \draw[thickedge,blue]
    (v1) .. controls ($(v1)!0.5!(c13a)$) .. (m13a);
  \draw[thickedge,red]
    (m13a) .. controls ($(c13a)!0.5!(v3)$) .. (v3);

  \draw[thickedge, red]
    (v1) .. controls ($(v1)!0.5!(c13b)$) .. (m13b);
  \draw[thickedge, blue]
    (m13b) .. controls ($(c13b)!0.5!(v3)$) .. (v3);

  \node[vert] at (v1) {1};
  \node[vert] at (v2) {2};
  \node[vert] at (v3) {3};
  \node[vert] at (v4) {4};
\end{tikzpicture}

\caption*{(a) A $4$-vertex $W$-cone with universal vertex labelled $1$ and monochromatic $K_3$}
\label{figconea}
\end{minipage}
\hfill
\begin{minipage}[t]{0.47\textwidth}
\centering
\begin{tikzpicture}[line cap=round, line join=round]
  \tikzset{
    vert/.style={circle, draw=black, fill=white,
                 minimum size=7mm, inner sep=0pt, font=\footnotesize},
    thickedge/.style={line width=1.1pt},
  }

  \coordinate (v1) at (0,0);
  \coordinate (v2) at (5,0);
      \coordinate (m13a) at (1.8,1.2); 

  \coordinate (v3) at (3.6,1.6);
  \coordinate (v4) at (6.4,1.6);
  \coordinate (v5) at (3.6,-1.6);
  \coordinate (v6) at (6.4,-1.6);

  \foreach \i in {2,3,4,5,6} {
    \coordinate (m1\i) at ($(v1)!0.5!(v\i)$);
    \draw[thickedge, blue] (v1) -- (m1\i);
    \draw[thickedge, red]   (m1\i) -- (v\i);
  }

  \draw[thickedge, blue] (v2) -- (v3);
  \draw[thickedge, blue] (v3) -- (v4);
  \draw[thickedge, blue] (v4) -- (v2);

  \draw[thickedge, blue] (v2) -- (v5);
  \draw[thickedge, blue] (v5) -- (v6);
  \draw[thickedge, blue] (v6) -- (v2);

    \draw[thickedge,red ]
    (v1) .. controls ($(v1)!0.5!(c13a)$) .. (m13a);
  \draw[thickedge,blue]
    (m13a) .. controls ($(c13a)!0.5!(v3)$) .. (v3);

  \node[vert] at (v1) {1};
  \node[vert] at (v2) {2};
  \node[vert] at (v3) {3};
  \node[vert] at (v4) {4};
  \node[vert] at (v5) {5};
  \node[vert] at (v6) {6};
\end{tikzpicture}

\caption*{(b) A $6$-vertex $W$-cone with universal vertex labelled $1$ and monochromatic 5-vertex Friendship graph $F_2$}
\end{minipage}
\caption{ Coloured examples of $W$-cones.}
\label{fig_cone}
\end{figure} 
\begin{restatable}{proposition}{WConeIsWState}
\label{wconiswstat}
W-cones are W-state graphs.     
\end{restatable}
\begin{proof}
The vertex condition of W-state graphs is satisfied due to \Cref{eq_vertex_condition}.
The matching condition follows from \Cref{eq_edge_partition} as every perfect matching can have exactly one edge incident on the universal vertex.
\end{proof}

Some examples of W-cones are shown in \Cref{fig_cone}.

\begin{restatable}{observation}{FactorCriticalEdgesNearPerfect}
\label{fac_maximal}
In a factor-critical graph, every edge is part of a near-perfect matching. 
\end{restatable}
\begin{proof}
Let \(F\) be a factor-critical graph and \(e=uv\in E(F)\) be an arbitrary edge. Since \(F\) is factor-critical, the graph \(F-u\) has a perfect matching \(M\).
Hence \(v\) is matched in \(M\), say to some vertex \(w\), so \(vw\in M\). Now replace the edge \(vw\) of \(M\) by \(uv\): \(M' := \bigl(M\setminus\{vw\}\bigr)\cup\{uv\}
\). Observe that \(M'\) is a near-perfect matching of \(F\) and contains the edge \(e\). The claim follows.
\end{proof}

From \cref{fac_maximal}, intuitively, a W‑cone is obtained by taking a factor‑critical graph $F$ and adding a new universal vertex $v$, making every edge incident with $v$ bichromatic (such that every vertex has a \red{red} half-edge incident on it), and keeping all edges within $F$ monochromatic \blue{blue}.

\begin{remark}\label{wcome_multi}
Observe that every W-cone has a multi-edge incident on its apex vertex.
\end{remark}

\subsection{Our Results: Characterization of W-state graphs}
Gu, Chen, Zeilinger and Krenn discovered that adding a universal vertex satisfying the vertex condition (\Cref{eq_vertex_condition}) to a monochromatic odd clique or a monochromatic friendship graph results in a W-state graph \cite{Quantum_graphs_3}. Note that both of these are W-cones. However, despite efforts to understand other constructions \cite{AIquantum2}, little was known. We resolve this by fully characterizing W-state graphs. 

\begin{restatable}{theorem}{mainthm}
\label{mainstructthm} A connected half-edge $2$-coloured matching-covered graph $(G,c)$ is a W-state graph if and only if
\begin{enumerate}
    \item Every vertex of $G$ has an incident red half-edge.
    \item All components of the standard $3$-connected decomposition of $G$, with the colouring inherited from $c$ on the real edges and virtual edges coloured bichromatically, are W-cones.
\end{enumerate}
\end{restatable}

\begin{remark}\label{aux_remark}
We emphasize that when we
decompose $G$ along a $2$-vertex cut $\{x,x'\}$, the induced subgraphs on each side need not themselves satisfy the vertex condition. 
The addition of virtual bichromatic edges between $x$ and $x'$ restores the vertex condition inside
each piece; this operation is formalized in
\Cref{def_W-block-decomp}. 
\end{remark}

Experimentally, a multi-edge requires a two-photon correlation with only one specific mode combination in the computational basis. It was not known whether one can create a W-state graph without this obstacle. Using our structural characterization, we resolve this.
\begin{restatable}{corollary}{NoSimpleWStateGraph}
\label{simplegraphcor}
No W-state graph is simple. Equivalently,
every W-state graph contains at least one pair of parallel edges.  
\end{restatable} 
Our structural characterization also gives an algorithm running as fast as verifying whether a graph is matching-covered. 

\begin{theorem}\label{thm_algorithm}
There exists an $O(|V|\cdot |E|)$-time algorithm that, given an uncoloured graph $G$, decides whether there exists a half-edge $2$-colouring $c$ such that $(G,c)$ is a W-state graph.
\end{theorem}
We also give a fast algorithm that can verify whether a given half-edge $2$-coloured graph $(G,c)$ is a W-state graph in $O(|V|\cdot |E|)$ time. 

Finally, we show the limitations of generalizing our approach to Dicke states, a generalization of W-states. We formally describe Dicke \textsc{FORALL-PMVC} in \Cref{dicke_section}. We prove that this problem is coNP-complete, thereby resolving an open problem posed by Vardi and Zhang~\cite{vardi2}.

\begin{restatable}{theorem}{DickeForallPMVCCoNPComplete}
\label{thm:dicke-forall-pmvc-conp-complete}
Dicke \textsc{FORALL-PMVC} is coNP-complete.
\end{restatable}

\subsection{Outline of the proof}
At a high level, the proof of \Cref{mainstructthm} proceeds in four steps, using some classical matching theory results along with the specific structure forced by the W-state conditions.

\smallskip
\noindent\textbf{Monochromatic edges form two factor--critical pieces.}
Given a W-state graph $(G,c)$, let $G_m$ be the spanning subgraph of $G$ consisting of all monochromatic (necessarily \blue{blue}) edges. Using the vertex condition (\Cref{eq_VC}), we show that every vertex in $G_m$ is inessential (\Cref{inessential_obs}). Using this along with Gallai-Edmonds decomposition, we show that $G_m$ consists of two connected components, each of which is factor-critical (in the case of W-cones, one of the components will be a trivial graph). Moreover, $E_b(G)=\delta(X,X')$ where $X$ and $X'$ are the vertex sets of the two connected components of $G_m$.

\smallskip
\noindent\textbf{3--connectivity and W--cones.}
Using the matching condition in \Cref{def_Wstate}, we get that $\delta(X,X')$ is a \emph{tight cut} (\Cref{def_tight}). If both $|X|$ and $|X'|$ are at least $3$, this cut is non-trivial, and one can show that $G$ is bicritical (\Cref{lem_G-bicrit}). By the Edmonds--Lovász--Pulleyblank characterization of bricks (\Cref{thm_brickchar}), a $3$-connected bicritical matching-covered graph is a brick and hence has no nontrivial tight cuts. This forces the following dichotomy:
\begin{itemize}
  \item either $G$ is not $3$--connected and has a $2$--vertex cut $\{x,x'\}$ with $x\in X$, $x' \in X'$ (\Cref{lem_2-cut});
  \item  or $G$ is $3$--connected, in which case one of $X,X'$ must have size $1$.
\end{itemize}
In the latter case, \Cref{lem_3conn-wcone} shows that every $3$--connected W-state graph is a W--cone.

\smallskip
\noindent\textbf{Decomposing along $2$--vertex cuts.}
For a general W-state graph that is not $3$--connected we work along a minimal $2$--vertex cut $\{x,x'\}$ with $x\in X$, $x'\in X'$. Removing $\{x,x'\}$ splits $G$ into components $C_1,\dots,C_t$. For each $C_i$ we form the induced subgraph on $V(C_i)\cup\{x,x'\}$ and ensure that $x$ and $x'$ are joined inside it by a bichromatic edge, adding one or, if needed to restore the vertex condition, two virtual parallel bichromatic edges between $x$ and $x'$. These graphs $(H_i,c_i)$ are the W-blocks of $(G,c)$ with respect to the cut (\Cref{def_W-block-decomp}). We show that each W-block is again a W-state graph (\Cref{thm_W-blocks-are-W}). Thus, by recursively decomposing along $2$-vertex cuts, we arrive at a collection of smaller, $3$-connected W-state graphs; by \Cref{lem_3conn-wcone}, each of these is a W-cone.

\smallskip 
\noindent\textbf{Glueing along bichromatic edges.} To prove the converse direction and obtain the full characterization, we introduce a natural glueing operation: the W-union of two W-state graphs along a pair of bichromatic edges (\Cref{def_W-union}). This operation identifies two bichromatic edges (possibly creating parallel edges between their endpoints) and preserves the W-state property (\Cref{lem_paste-over-edge}). Since every block in the above decomposition contains a bichromatic edge between the two vertices of the adhesion pair, the standard decomposition of a graph into $3$-connected components with virtual edges can be realized as a tree of such components glued by W-unions. Hence, any matching-covered graph whose $3$-connected components (with these virtual bichromatic edges) are W-cones, and in which the vertex condition is satisfied, can be constructed by iterated W-unions of W-cones and is therefore a W-state graph. This completes the proof of \Cref{mainstructthm}. Using \Cref{mainstructthm}, we prove \Cref{simplegraphcor} and \Cref{thm_algorithm}.



\section{Merging and Decomposing W-state graphs}
\subsection{Matching theory preliminaries}

\begin{definition}[Tight cut]\label{def_tight}
In a matching-covered graph \(G\), a non-empty proper subset
\(X\subsetneq V(G)\) is called \emph{tight} if
\(\lvert M\cap\delta(X)\rvert = 1\) for all $M \in \mathcal{PM}(G)$.
\end{definition}

If $|X|=1$, then $\delta(X)$ is called the trivial tight cut. Recursively contracting both shores of every non-trivial tight cut
produces the \emph{tight-cut decomposition} of a matching-covered
graph into \emph{bricks} (non-bipartite) and \emph{braces}
(bipartite). In their seminal paper, Lovász and Plummer proved that the tight-cut decomposition is unique up to the ordering of its components~\cite{LovaszPlummer}.

\begin{definition}[Brick]\label{def_brick}
A \emph{brick} is a non-bipartite matching-covered graph that has no
non-trivial tight cuts.
\end{definition}

\begin{definition}[Bicritical graph]\label{def_bicrit}
A graph \(G\) with \(|V(G)|\ge4\) is \emph{bicritical} if
\(G-\{u,v\}\) has a perfect matching for every distinct
\(u,v\in V(G)\).
\end{definition}

Using LP duality, Edmonds, Lovász and Pulleyblank \cite{DBLP:journals/combinatorica/EdmondsLP82} proved that bricks are precisely the 3-connected bicritical graphs.
\begin{theorem}[{\cite{DBLP:journals/combinatorica/EdmondsLP82, DBLP:journals/jct/Lovasz87, DBLP:journals/combinatorica/Szigeti02}}]\label{thm_brickchar}
A graph is a brick if and only if it is $3$-connected and bicritical.
\end{theorem}

\begin{definition}
A vertex is defined to be \emph{essential} if it is covered by all maximum matchings. Otherwise, it is defined to be \emph{inessential}.   
\end{definition}

The following fact follows from the Gallai–Edmonds decomposition.

\begin{fact}\label{egdecom}
If every vertex of a graph is inessential, then each connected component is factor-critical. 
\end{fact}

\begin{restatable}{observation}{MatchingCoveredNoCutVertex}
\label{cut_ver_mcg}
A connected matching-covered graph $G$ has no cut vertex. In particular, $G$
is $2$-connected.
\end{restatable}
\begin{proof}
Suppose, for contradiction, that $v$ is a cut vertex in $G$. Let $C_1,\dots,C_k$ ($k \ge 2$) be the components of
$G - \{v\}$. For each $i$, choose a neighbour $x_i \in V(C_i)$ of $v$; since
$G$ is matching-covered, the edge $vx_i$ lies in some perfect matching $M_i$.
In $M_i$, all vertices of $C_j$ ($j \ne i$) are matched within $C_j$, so
$|V(C_j)|$ is even, while $V(C_i) \setminus \{x_i\}$ is perfectly matched
inside $C_i$, so $|V(C_i)|$ is odd. Choosing two different indices
$i \ne j$ yields a parity contradiction. Thus, $G$ has no cut vertex and is
$2$-connected.
\end{proof}

\subsection{Decomposing into smaller W-state graphs}
Throughout the rest of the paper, we fix a W-state graph \((G,c)\) and write \(n=|V(G)|\). Let \(G_m\coloneqq (V,E_m(G))\). Note that every edge of \(G_m\) is \blue{blue}. Using \cref{egdecom}, we prove below that $G_m$ has exactly two components, both of which are factor-critical (and hence odd).

\begin{restatable}{observation}{obsmm}
\label{obs_mm-size}
The cardinality of a maximum matching in \(G_m\) is \(n/2-1\).
\end{restatable}
\begin{proof}
If \(G_m\) had a perfect matching, \(G\) would have a perfect matching containing no bichromatic edge, contradicting
the definition of $W$-state graphs. Conversely, any perfect matching of \(G\) becomes a matching of size \(n/2-1\) in \(G_m\) once its unique bichromatic edge is removed. Hence, the maximum matching in $G_m$ has size \(n/2-1\).
\end{proof}

\begin{restatable}{observation}{EveryVertexInessential}
\label{inessential_obs}
Every vertex in $G_m$ is inessential.     
\end{restatable}
\begin{proof}
Pick an arbitrary vertex $v$. Since $G$ is a W-state graph, there must be a bichromatic edge $e \in E(G)$ incident on $v$, and there exists a perfect matching $M$ containing $e$. From \Cref{obs_mm-size}, observe that $M\setminus \{e\}$ is a maximum matching of $G_m$ which exposes $v$, implying that $v$ is inessential.
\end{proof}

\begin{restatable}{lemma}{MonochromaticGraphTwoFactorCriticalComponents}
\label{thm_two-fc}
The graph \(G_m\) is the disjoint union of exactly two factor-critical components.
\end{restatable}
\begin{proof}
From \Cref{inessential_obs}, every vertex of \(G_m\) is inessential; hence, by
\Cref{egdecom}, every component of \(G_m\) is factor-critical. If there are $c$ such factor-critical components $C_i$, then a maximum matching in $G_m$ leaves exactly one exposed vertex in each component, so the size of the maximum matching is $ \sum_i \frac{|C_i|-1}{2}=\frac{n-c}{2}.$ From \cref{obs_mm-size}, this should be $n/2-1$. Therefore, $c=2$ and hence \(G_m\) consists of two (necessarily odd) factor-critical components whose union is $G_m$.
\end{proof}

\begin{restatable}{observation}{facdeco}
\label{facdeco}
$E_b(G)$ is the cut between the two factor critical components of $G_m$.
\end{restatable}
\begin{proof} 
By \Cref{thm_two-fc}, we may partition the vertex set \(V(G)=X\sqcup X'\) so that \(G[X]\) and \(G[X']\) are factor-critical and  $
  E_m(G)\;=\;E\bigl(G[X]\bigr)\sqcup E\bigl(G[X']\bigr)$. 
Towards a contradiction, say there is a bichromatic edge $e$ with both endpoints in $X$. Let $M$ be a perfect matching containing $e$. As $|X|$ is odd, there must exist a vertex in $X$ matching to a vertex in $X'$ in $M$. Such an edge must be bi-chromatic. This produces two bichromatic edges $M$, which is a contradiction. Therefore, $ 
  E_b(G)\;=\;\delta\bigl(X,X'\bigr)$.
\end{proof}

By \Cref{thm_two-fc} and \Cref{facdeco}, we may partition the vertex set
\(V(G)=X\sqcup X'\) so that \(G[X]\) and \(G[X']\) are factor-critical,  $
  E_m(G)\;=\;E\bigl(G[X]\bigr)\sqcup E\bigl(G[X']\bigr)$ and $ 
  E_b(G)\;=\;\delta\bigl(X,X'\bigr)$. 
  
We now prove that $G$ has a $2$-vertex cut or is a W-cone.

\begin{restatable}{observation}{SingletonShoreImpliesWCone}
\label{wcone_building}
If $|X|=1$ or $|X'|=1$, then $(G,c)$ is a W-cone. 
\end{restatable}
\begin{proof}
Without loss of generality, let $|X|=1$ and $X=\{x\}$.
Since every vertex in $X'$ must have an incident bichromatic edge whose other endpoint is $x$, $x$ is a universal vertex.
Therefore, the edge partition condition of \Cref{conedef} follows from \Cref{thm_two-fc}. The vertex condition holds as it is the same condition as for W-states.
\end{proof}

%


\begin{restatable}{observation}{NontrivialTightCutPreventsBrick}
\label{obs_tight-cut}
If $\min(|X|,|X'|)\geq 3$, then $G$ is not a brick. 
\end{restatable}
\begin{proof}
Recall that $E_b(G)\;=\;\delta\bigl(X,X'\bigr)$. By the definition of a W-state graph, we know that $|M \cap E_b(G)|=1$ for all $M \in \mathcal{PM}(G)$. Therefore, $\delta\bigl(X,X'\bigr)$ is a tight cut. Moreover, this is non-trivial as $|X|,|X'|$ are at least $3$. Therefore, from \Cref{def_brick}, $G$ is not a brick.
\end{proof}

\begin{restatable}{observation}{LargeShoreImpliesBicritical}
\label{lem_G-bicrit}
If $\min(|X|,|X'|)\geq 3$, then \(G\) is bicritical.
\end{restatable}
\begin{proof}
Let $\{u,v\} \in V(G)=V$ be any two distinct vertices.
We will show that $G[V\setminus \{u, v\}]$ has a perfect matching.

\smallskip
\noindent\emph{Case 1: \(u\in X\) and \(v\in X'\).}
Since \(G[X]\) and \(G[X']\) are factor-critical there exist perfect
matchings \(M_X\) of \(G[X\setminus\{u\}]\) and \(M_{X'}\) of \(G[X'\setminus\{v\}]\).
Their union \(M_X\cup M_{X'}\) is a perfect matching of $G[V\setminus \{u, v\}]$.

\smallskip
\noindent\emph{Case 2: \(u,v\in X\) (the case \(u,v\in X'\) is
symmetric).}
Choose a perfect matching \(M_X\) of \(G[X\setminus\{u\}]\) and let \(e\in M_X\)
be the edge incident with \(v\). Let $v'$ be the other endpoint of $e$.
Let \(v'v''\) be a bichromatic edge whose \red{red} half-edge is incident on \(v'\).
Because \(G[X']\) is factor-critical, \(G[X'\setminus\{v''\}]\) has a perfect
matching \(M_{X'}\).
The set \((M_X\setminus\{e\})\cup\{v'v''\}\cup M_{X'}\) is a perfect
matching of $G[V\setminus \{u, v\}]$
\end{proof}

Using \cref{thm_brickchar}, along with \cref{obs_tight-cut} and \cref{lem_G-bicrit} we can prove that W-state graphs have a very specific 2-vertex cut. 

\begin{restatable}{lemma}{SpecificTwoVertexCut}
\label{lem_2-cut}
If $\min(|X|,|X'|)\geq 3$, then the graph \(G\) has a two-vertex cut \(\{x,x'\}\) with
\(x\in X\) and \(x'\in X'\).
\end{restatable}
\begin{proof}
Towards a contradiction, let us assume that \(G\) is $3$-connected.
By \Cref{lem_G-bicrit}, \(G\) is bicritical. Hence, by
\Cref{thm_brickchar}, \(G\) is a brick, contradicting
\Cref{obs_tight-cut}.
Thus \(G\) is not $3$-connected; let \(\{x,x'\}\) be a minimal
vertex cut.

Towards a contradiction suppose that $\{x,x'\} \subseteq X$.
Since $|X|\geq 3$, $X\setminus \{x,x'\}$ is non-empty. Observe that every vertex in $X\setminus \{x,x'\}$ has a bichromatic edge incident on it with the other end-point in $X'$. Therefore, every vertex in $X-\{x,x'\}$ is connected to some vertex in $X'$. From \Cref{thm_two-fc}, we know that $G[X']$ is connected. Therefore, $G$ remains connected after the removal of $\{x,x'\}$. This contradicts the fact that $\{x,x'\}$ is a vertex cut. Therefore, $\{x,x'\} \not\subseteq X$. Similarly,  $\{x,x'\} \not\subseteq X'$. Therefore, one of the cut vertices lies in \(X\) and the other in \(X'\).
\end{proof}

\begin{restatable}{observation}{ThreeConnectedImpliesWCone}
\label{lem_3conn-wcone}
If $G$ is $3$-connected, then $(G,c)$ is a W-cone.
\end{restatable}
\begin{proof}
If both $\lvert X\rvert$ and $\lvert X'\rvert$ were at least $3$, then from \Cref{lem_2-cut}, we get that $G$ has a $2$-vertex cut, contradicting the
assumption that $G$ is $3$-connected. Hence at least one of $|X|,|X'|$ is $1$ in which case \Cref{wcone_building}
implies that $(G,c)$ is a W-cone.
\end{proof}

It follows that \Cref{mainstructthm} is true when $|X|=1$  or $|X'|=1$ (or when $G$ is $3$-connected), and it is sufficient to prove \Cref{mainstructthm} when $|X|$ and $|X'|$ are at least $3$.

We are now ready to describe the decomposition over the {2}-vertex cuts. 
\begin{definition}[Block decomposition of a W-state graph]\label{def_W-block-decomp}

Let $(G, c)$ be a W-state graph with partition $V(G)=X \sqcup X'$ as in \Cref{thm_two-fc}, and let $\{x, x'\}$ be the two-vertex cut given by \Cref{lem_2-cut}, where $x \in X$ and $x' \in X'$. Let $C_1, \ldots, C_t$ be the connected components of $G-\{x,x'\}$. For each $i \in [t]$, let $\widehat{H}_i := G[V(C_i)\cup\{x,x'\}]$. We obtain $(H_i,c_i)$ by initializing $H_i = \widehat{H}_i$ and picking a colouring $c_i$ which agrees with $c$ on the original half-edges in $G$ and making the following changes:
\begin{enumerate}
    \item If $x$ (respectively $x'$) does not have a red half-edge incident on it in $(H_i,c_i)$, we add a red-blue edge between $x$ and $x'$ with the red half-edge incident on $x$ (respectively $x'$).
    \item If ${H}_i$ still does not contain a bichromatic edge between $x$ and $x'$, we add a virtual red-blue edge between $x$ and $x'$ (with the red half-edge incident on either).
\end{enumerate}
We call $(H_i,c_i)$ the \emph{W-blocks} of $(G,c)$ with respect to the cut $\{x,x'\}$.
\end{definition}

\begin{remark}
This process may create a new multi-edge between $x$ and $x'$, and in general may create up to two virtual parallel bichromatic edges between $x$ and $x'$.
\end{remark}

\begin{figure}[ht] 
\centering
\begin{tikzpicture}[line cap=round, line join=round]
  \tikzset{
    vert/.style={circle, draw=black, fill=white,
                 minimum size=7mm, inner sep=0pt, font=\footnotesize},
    thickedge/.style={line width=1.1pt},
  }

  \coordinate (v1) at (0,0);
  \coordinate (v2) at (-2,0);
  \coordinate (v3) at (-5.8,1.6);
  \coordinate (v4) at (-5.8,-1.6);
    \coordinate (u3) at (3.8,1.6);
  \coordinate (u4) at (3.8,-1.6);

  \coordinate (m12) at ($ (v1)!0.5!(v2) $);

  \coordinate (c13a) at (-2.5,1.9);
  \coordinate (c13b) at (-2.5,0.9);
  \coordinate (m13a) at (-2.7,1.35);
  \coordinate (m13b) at (-2.7,0.85);

\coordinate (f12a) at (-1,0.15);

\coordinate (f12aL) at ($0.5*(v1) + 0.5*(f12a)$);
\coordinate (f12aR) at ($0.5*(f12a) + 0.5*(v2)$);
\coordinate (g12a)  at ($0.5*(f12aL) + 0.5*(f12aR)$);


  \coordinate (m14) at ($ (v1)!0.5!(v4) $);
  \draw[thickedge, blue] (v1) -- (m14);
  \draw[thickedge, red]   (m14) -- (v4);

    \coordinate (n14) at ($ (v2)!0.5!(u4) $);
  \draw[thickedge, blue] (v2) -- (n14);
  \draw[thickedge, red]   (n14) -- (u4);

  \draw[thickedge, blue] (v2) -- (v3);
  \draw[thickedge, blue] (v2) -- (v4);
  \draw[thickedge, blue] (v3) -- (v4);

    \draw[thickedge, blue] (v1) -- (u3);
  \draw[thickedge, blue] (v1) -- (u4);
  \draw[thickedge, blue] (u3) -- (u4);

  \draw[thickedge,blue]
    (v1) .. controls ($(v1)!0.5!(c13a)$) .. (m13a);
  \draw[thickedge,red]
    (m13a) .. controls ($(c13a)!0.5!(v3)$) .. (v3);

  \draw[thickedge, red]
    (v1) .. controls ($(v1)!0.5!(c13b)$) .. (m13b);
  \draw[thickedge, blue]
    (m13b) .. controls ($(c13b)!0.5!(v3)$) .. (v3);

  \coordinate (q13a) at (0.5,1.9);
  \coordinate (q13b) at (0.5,0.9);
  \coordinate (n13a) at (0.7,1.35);
  \coordinate (n13b) at (0.7,0.85);
  \draw[thickedge,blue]
    (v2) .. controls ($(v2)!0.5!(q13a)$) .. (n13a);
  \draw[thickedge,red]
    (n13a) .. controls ($(q13a)!0.5!(u3)$) .. (u3);

  \draw[thickedge, red]
    (v2) .. controls ($(v2)!0.5!(q13b)$) .. (n13b);
  \draw[thickedge, blue]
    (n13b) .. controls ($(q13b)!0.5!(u3)$) .. (u3);

  \node[vert] at (v1) {$22'$};
  \node[vert] at (v2) {$11'$};
  \node[vert] at (v3) {$3$};
  \node[vert] at (v4) {$4$};
    \node[vert] at (u3) {$3'$};
  \node[vert] at (u4) {$4'$};
\end{tikzpicture}
\caption{A $6$-vertex W-state graph. Vertices $11'$ and $22'$ act as the $2$-vertex cut and the decomposition along this cut gives two new W-state graphs shown in \Cref{ex_decomp}.}
\label{ex_wout_edge}
\end{figure}
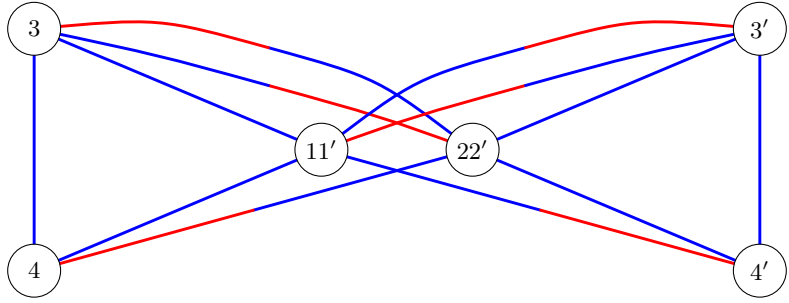

\begin{figure}[ht]
\centering

\begin{minipage}[t]{0.47\textwidth}
\centering
\begin{tikzpicture}[line cap=round, line join=round]
  \tikzset{
    vert/.style={circle, draw=black, fill=white,
                 minimum size=7mm, inner sep=0pt, font=\footnotesize},
    thickedge/.style={line width=1.1pt},
  }

  \coordinate (v1) at (0,0);
  \coordinate (v2) at (-2,0);
  \coordinate (v3) at (-5.8,1.6);
  \coordinate (v4) at (-5.8,-1.6);

  \coordinate (m12) at ($ (v1)!0.5!(v2) $);
  \coordinate (m14) at ($ (v1)!0.5!(v4) $);

  \coordinate (c13a) at (-2.5,1.9);
  \coordinate (c13b) at (-2.5,0.9);
  \coordinate (m13a) at (-2.7,1.35);
  \coordinate (m13b) at (-2.7,0.85);

  \draw[thickedge, blue] (v1) -- (m12);
  \draw[thickedge, red]   (m12) -- (v2);

  \draw[thickedge, blue] (v1) -- (m14);
  \draw[thickedge, red]   (m14) -- (v4);

  \draw[thickedge, blue] (v2) -- (v3);
  \draw[thickedge, blue] (v2) -- (v4);
  \draw[thickedge, blue] (v3) -- (v4);

  \draw[thickedge,blue]
    (v1) .. controls ($(v1)!0.5!(c13a)$) .. (m13a);
  \draw[thickedge,red]
    (m13a) .. controls ($(c13a)!0.5!(v3)$) .. (v3);

  \draw[thickedge, red]
    (v1) .. controls ($(v1)!0.5!(c13b)$) .. (m13b);
  \draw[thickedge, blue]
    (m13b) .. controls ($(c13b)!0.5!(v3)$) .. (v3);

  \node[vert] at (v1) {2};
  \node[vert] at (v2) {1};
  \node[vert] at (v3) {3};
  \node[vert] at (v4) {4};
\end{tikzpicture}

\caption*{(a) A $W$-cone with monochromatic $K_3$}
\end{minipage}
\hfill
\begin{minipage}[t]{0.47\textwidth}
\centering
\begin{tikzpicture}[line cap=round, line join=round]
  \tikzset{
    vert/.style={circle, draw=black, fill=white,
                 minimum size=7mm, inner sep=0pt, font=\footnotesize},
    thickedge/.style={line width=1.1pt},
  }

  \coordinate (v1) at (0,0);
  \coordinate (v2) at (2,0);
  \coordinate (v3) at (5.8,1.6);
  \coordinate (v4) at (5.8,-1.6);

  \coordinate (m12) at ($ (v1)!0.5!(v2) $);
  \coordinate (m14) at ($ (v1)!0.5!(v4) $);

  \coordinate (c13a) at (2.5,1.9);
  \coordinate (c13b) at (2.5,0.9);
  \coordinate (m13a) at (2.7,1.35);
  \coordinate (m13b) at (2.7,0.85);

  \draw[thickedge, blue] (v1) -- (m12);
  \draw[thickedge, red]   (m12) -- (v2);

  \draw[thickedge, blue] (v1) -- (m14);
  \draw[thickedge, red]   (m14) -- (v4);

  \draw[thickedge, blue] (v2) -- (v3);
  \draw[thickedge, blue] (v2) -- (v4);
  \draw[thickedge, blue] (v3) -- (v4);

  \draw[thickedge,blue]
    (v1) .. controls ($(v1)!0.5!(c13a)$) .. (m13a);
  \draw[thickedge,red]
    (m13a) .. controls ($(c13a)!0.5!(v3)$) .. (v3);

  \draw[thickedge, red]
    (v1) .. controls ($(v1)!0.5!(c13b)$) .. (m13b);
  \draw[thickedge, blue]
    (m13b) .. controls ($(c13b)!0.5!(v3)$) .. (v3);

  \node[vert] at (v1) {$1'$};
  \node[vert] at (v2) {$2'$};
  \node[vert] at (v3) {$3'$};
  \node[vert] at (v4) {$4'$};
\end{tikzpicture}

\caption*{(b) A $W$-cone with monochromatic $K_3$}
\end{minipage}
\caption{Block decomposition of the W-state in \Cref{ex_wout_edge}}
\label{ex_decomp}
\end{figure}

\begin{remark}\label{ver_cond_rem}
The bichromatic edge(s) between $x$ and $x'$ in \Cref{def_W-block-decomp} play two roles. First, they ensure that each W-block satisfies the vertex condition. Second, they make the W-union of \Cref{def_W-union} the inverse of the block decomposition. The original graph $(G,c)$ is recovered from the collection of W-blocks by pasting them along $\{x,x'\}$ and deleting precisely those virtual bichromatic edges that were not present in $G$.
\end{remark}

\begin{restatable}[W-blocks are W-state graphs]{lemma}{WBlocksAreWStateGraphs}
\label{thm_W-blocks-are-W}
Let $(G,c)$ be a W-state graph and let
$
( H_1, c_1),\dots,( H_t,c_t)$
be the W-blocks obtained from $(G,c)$ by the block decomposition of \Cref{def_W-block-decomp}. Then each $( H_i, c_i)$ is a W-state graph.
\end{restatable}
\begin{proof}
We use $H,C$ to represent $H_i,C_i$, respectively for some arbitrary $i \in [t]$. Let $Y \coloneqq  (X \cap V(C)) \cup \{x\}$ and $Y' \coloneqq  (X' \cap V(C)) \cup \{x'\}$,
so $V(H) = Y \sqcup Y'$.
Observe that the bichromatic edges between $x,x'$ in $H$ are between vertices in $Y$ and $Y'$. All the other edges in $H$ are contained in $G$ and in $G$ all
monochromatic edges lie inside $X$ or inside $X'$, while all bichromatic edges lie between
$X$ and $X'$. So, we have
\[
E_m(H) = E(H[Y]) \sqcup E(H[Y']) \quad\text{and}\quad E_b(H) = \delta(Y,Y').
\]
Recall that $G[X]$ and $G[X']$ are factor-critical graphs. On removal of $x$ (respectively $x'$) from $G[X]$ (respectively $G[X']$), we see that there exists a perfect matching on each of the connected components of $G[X\setminus \{x\}]$ (respectively $G[X'\setminus \{x'\}]$). It is now easy to see the following perfect matchings:
\begin{itemize}
    \item $P_C$: A matching which covers the $V(C)$.
    \item $P^c_{H}$: A matching which covers $V(G)\setminus \ V(H)$.
    \item $P^c_{C}$: A matching which covers $V(G)\setminus V(C)$.
\end{itemize}
It follows that $|Y|$ and $|Y'|$ must be odd.

\medskip
\noindent\emph{(I) $H$ is matching-covered}. Consider any edge $e \in E(H)$. If $e$ is between $x$ and $x'$, it is part of the perfect matching $P_C \cup \{e\}$ on $H$. If $e$ is not between $x$ and $x'$, it must have been part of the matching-covered graph $G$. Therefore, there exists a perfect matching $P$ on $G$ containing $e$. As $|Y|$ and $|Y'|$ have the same parity, $P$ must either have both the vertices $x,x'$ matched to vertices in $H$ or both the vertices $x,x'$ matched to vertices outside $H$. If it is of the former type $P \cap E(H)$ is a perfect matching on $H$. If it is the latter type, then $(P \cap E(H)) \cup \{e'\}$ is a perfect matching on $H$ where $e'$ is an edge between $x,x'$.

\medskip
\noindent\emph{(II) $H$ satisfies the vertex condition.} $x,x'$ satisfy the vertex condition by construction (see \Cref{ver_cond_rem}). Every vertex of $H$ apart from $x,x'$ has the same incident
edges and half-edge colours as in $G$, and $(G,c)$ satisfies the vertex
condition. Therefore, the vertex condition is satisfied for all vertices in $H$.

\medskip
\noindent\emph{(III) $H$ satisfies the matching condition.}
Now let $M$ be an arbitrary perfect matching of $H$, and set
$k \coloneqq  |M \cap E_b(H)|$. Counting vertices of $Y$ covered by $M$, we see that
monochromatic edges in $E(H[Y])$ cover $2$ vertices of $Y$ and each bichromatic edge
covers $1$ vertex of $Y$. Since $|Y|$ is odd, $k$ must be odd; in particular $k \ge 1$. Thus every perfect matching of $H$
contains at least one bichromatic edge.

We claim that in fact $k \le 1$. Suppose, towards a contradiction, that there
exist $M \in \mathcal{PM}(H)$ with $k \ge 3$.

We construct $M'$ by tweaking $M$, depending on whether there is an edge between $x,x'$ in $M$.
\begin{enumerate}
    \item Suppose $M$ does not contain an edge between $x,x'$. Then, define $M' \coloneqq  M \;\cup\; P^c_H$. 
    \item Suppose $M$ contains an edge $e$ between $x,x'$. Then define $M' \coloneqq  M \setminus \{e\} \;\cup\; P^c_C$.
\end{enumerate}
It is easy to see that $M'$ is a perfect matching of $G$. As all bichromatic edges of $M$ except for $e$ are also part of $M'$, we get that $|M' \cap E_b(G)|\geq 2$, which contradicts the
W-state property of $(G,c)$.

Thus no such $M$ can exist, and for every $M \in \mathcal{PM}(H)$ we have
$|M \cap E_b(H)| \le 1$. Combined with the parity argument ($k$ is odd), this gives
$|M \cap E_b(H)| = 1$ for all perfect matchings $M$ of $H$. Therefore, $H$ is a W-state graph.
\end{proof}

\subsection{Merging two W-states}
\begin{definition}[W-union along a bichromatic edge]\label{def_W-union}
Let $(G_1,c_1)$ and $(G_2,c_2)$ be W-state graphs, and let
$e_1 = x_1x_1' \in E_b(G_1)$ and $e_2 = x_2x_2' \in E_b(G_2)$ be bichromatic
edges.

\begin{enumerate}[(i)]
  \item We first form the graph $G$ by identifying $x_1$ with $x_2$ into a
  new vertex $x$, and $x_1'$ with $x_2'$ into a new vertex $x'$, and retaining
  the edges $e_1,e_2$ between $x$ and $x'$ (possibly becoming a multi-edge). The colouring on all half-edges is inherited
  from $c_1$ and $c_2$. We call $(G,c)$ the \emph{retained W-union} of
  $(G_1,c_1)$ and $(G_2,c_2)$ along $e_1,e_2$.
  \item  Optionally, we can delete the edge(s) between the vertices $x,x'$, if there is still a \red{red} half-edge incident on them after the removal. 
\end{enumerate}
We will refer to both $(G,c)$ and the graph obtained after the edge-deletions as in (ii) simply as \emph{W-unions} of $(G_1,c_1)$ and $(G_2,c_2)$ along $e_1,e_2$, specifying
“retained” when we need to emphasise that both $e_1,e_2$ are present.
\end{definition}
The retained W-union of the graphs in \Cref{ex_decomp} would give the graph in \Cref{ex_with_edge}. After applying edge deletion as per (ii), we get the graph in \Cref{ex_wout_edge}.

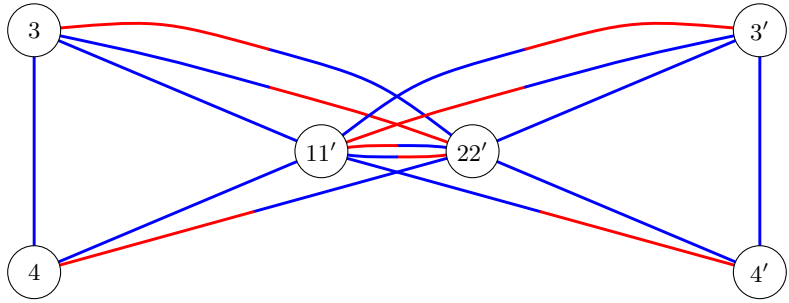
\begin{figure}[ht]
\centering
\begin{tikzpicture}[line cap=round, line join=round]
  \tikzset{
    vert/.style={circle, draw=black, fill=white,
                 minimum size=7mm, inner sep=0pt, font=\footnotesize},
    thickedge/.style={line width=1.1pt},
  }

  \coordinate (v1) at (0,0);
  \coordinate (v2) at (-2,0);
  \coordinate (v3) at (-5.8,1.6);
  \coordinate (v4) at (-5.8,-1.6);
    \coordinate (u3) at (3.8,1.6);
  \coordinate (u4) at (3.8,-1.6);

  \coordinate (m12) at ($ (v1)!0.5!(v2) $);

  \coordinate (c13a) at (-2.5,1.9);
  \coordinate (c13b) at (-2.5,0.9);
  \coordinate (m13a) at (-2.7,1.35);
  \coordinate (m13b) at (-2.7,0.85);

\coordinate (f12a) at (-1,0.15);

\coordinate (f12aL) at ($0.5*(v1) + 0.5*(f12a)$);
\coordinate (f12aR) at ($0.5*(f12a) + 0.5*(v2)$);
\coordinate (g12a)  at ($0.5*(f12aL) + 0.5*(f12aR)$);

\draw[thickedge, blue]
  (v1)  .. controls (f12aL) .. (g12a);
\draw[thickedge, red]
  (g12a) .. controls (f12aR) .. (v2);

\coordinate (f12b) at (-1,-0.15);

\coordinate (f12bL) at ($0.5*(v1) + 0.5*(f12b)$);
\coordinate (f12bR) at ($0.5*(f12b) + 0.5*(v2)$);
\coordinate (g12b)  at ($0.5*(f12bL) + 0.5*(f12bR)$);

\draw[thickedge, red]
  (v1)  .. controls (f12bL) .. (g12b);
\draw[thickedge, blue]
  (g12b) .. controls (f12bR) .. (v2);

  \coordinate (m14) at ($ (v1)!0.5!(v4) $);
  \draw[thickedge, blue] (v1) -- (m14);
  \draw[thickedge, red]   (m14) -- (v4);

    \coordinate (n14) at ($ (v2)!0.5!(u4) $);
  \draw[thickedge, blue] (v2) -- (n14);
  \draw[thickedge, red]   (n14) -- (u4);

  \draw[thickedge, blue] (v2) -- (v3);
  \draw[thickedge, blue] (v2) -- (v4);
  \draw[thickedge, blue] (v3) -- (v4);

    \draw[thickedge, blue] (v1) -- (u3);
  \draw[thickedge, blue] (v1) -- (u4);
  \draw[thickedge, blue] (u3) -- (u4);

  \draw[thickedge,blue]
    (v1) .. controls ($(v1)!0.5!(c13a)$) .. (m13a);
  \draw[thickedge,red]
    (m13a) .. controls ($(c13a)!0.5!(v3)$) .. (v3);

  \draw[thickedge, red]
    (v1) .. controls ($(v1)!0.5!(c13b)$) .. (m13b);
  \draw[thickedge, blue]
    (m13b) .. controls ($(c13b)!0.5!(v3)$) .. (v3);

  \coordinate (q13a) at (0.5,1.9);
  \coordinate (q13b) at (0.5,0.9);
  \coordinate (n13a) at (0.7,1.35);
  \coordinate (n13b) at (0.7,0.85);
  \draw[thickedge,blue]
    (v2) .. controls ($(v2)!0.5!(q13a)$) .. (n13a);
  \draw[thickedge,red]
    (n13a) .. controls ($(q13a)!0.5!(u3)$) .. (u3);

  \draw[thickedge, red]
    (v2) .. controls ($(v2)!0.5!(q13b)$) .. (n13b);
  \draw[thickedge, blue]
    (n13b) .. controls ($(q13b)!0.5!(u3)$) .. (u3);

  \node[vert] at (v1) {$22'$};
  \node[vert] at (v2) {$11'$};
  \node[vert] at (v3) {3};
  \node[vert] at (v4) {4};
    \node[vert] at (u3) {$3'$};
  \node[vert] at (u4) {$4'$};
\end{tikzpicture}
\caption{A retained W-union of the graphs in \Cref{ex_decomp}}
\label{ex_with_edge}
\end{figure}

\begin{restatable}[Pasting two W-states along a bichromatic edge]{lemma}{PastingTwoWStatesAlongBichromaticEdge}
\label{lem_paste-over-edge}
Let $(G_1,c_1)$ and $(G_2,c_2)$ be W-state graphs, and let
$e_1\in E_b(G_1)$ and $e_2\in E_b(G_2)$ be bichromatic edges. Then every
W-union of $(G_1,c_1)$ and $(G_2,c_2)$ along $e_1,e_2$ is a W-state graph.
\end{restatable}
\begin{proof}
From \Cref{thm_two-fc}, let the vertex partitions of the  W-state graphs $(G_1,c_1)$ and $(G_2,c_2)$ be $V(G_i)=X_i \sqcup X_i'$ $(i=1,2)$ such that
$E_m(G_i)=E(G_i[X_i]) \sqcup E(G_i[X_i'])$ and
$E_b(G_i)=\delta(X_i,X_i')$.

Let $x_1\in X_1$, $x_1'\in X_1'$ and $x_2\in X_2$, $x_2'\in X_2'$.
Let $G$ be formed by the W-union as described in \Cref{def_W-union}. Let $c$ be the induced half-edge colouring on $G$. We now prove that $(G,c)$ is a W-state graph. Set
\[
Y_1 \coloneqq  X_1 \setminus \{x_1\},\quad
Y_1' \coloneqq  X_1' \setminus \{x_1'\},\quad
Y_2 \coloneqq  X_2 \setminus \{x_2\},\quad
Y_2' \coloneqq  X_2' \setminus \{x_2'\},
\]
and define
\[
X \coloneqq  Y_1 \cup Y_2 \cup \{x\},\qquad
X' \coloneqq  Y_1' \cup Y_2' \cup \{x'\}.
\]
Because we only identify vertices and do not create any new edges, every edge that was monochromatic in $G_1$ or $G_2$ still has both
endpoints in $X$ or both in $X'$, and every edge that was bichromatic still
has one endpoint in $X$ and one in $X'$. Hence,
\[
E_m(G) = E(G[X]) \sqcup E(G[X']) \qquad\text{and}\qquad
E_b(G) = \delta(X,X').
\]

Since $G_1[X_1]$ and $G_2[X_2]$ are factor-critical and $X$ is obtained by
taking a vertex-sum of $G_1[X_1]$ and $G_2[X_2]$ at $x$, the graph $G[X]$ is
also factor-critical. Similarly $G[X']$ is factor-critical. It follows that $|X|$ and $|X'|$ are odd.

\medskip
\noindent\emph{(I) $G$ is matching-covered}.
Let $M_2$ be a perfect matching of $G_2$ that contains $e_2$ (which exists
because $(G_2,c_2)$ is matching-covered). As $G_2$ is connected, there exists another edge $e'$ which is incident on the same vertex as $e_2$. Consider a perfect matching $M_3$ on $G_2$ containing $e'$. Observe that $M_3$ does not contain $e_2$. If $M$ is an arbitrary perfect matching of
$G_1$, then:
\begin{itemize}
  \item if $e_1 \in M$, then $M' \coloneqq  (M \setminus \{e_1\}) \,\cup\, M_3$
  is a perfect matching of $G$ containing the edges of $M$.
  \item if $e_1 \notin M$, then $M$ covers $x_1$ and $x_1'$ by edges
  other than $e_1$, and
$
    M' \coloneqq  M \,\cup\, (M_2 \setminus \{e_2\})
$
  is a perfect matching of $G$ containing the edges of $M$.
\end{itemize}
Since every edge of $G_1$ lies in some perfect matching of $G_1$, it follows that every edge of $G_1$ also lies in some perfect matching of $G$. By symmetry, every edge of $G_2$ lies in some perfect matching of $G$. Finally, any retained edge $e$ between $x$ and $x'$ lies
in a perfect matching $(M_1 \setminus \{e_1\}) \,\cup\, (M_2 \setminus \{e_2\}) \,\cup\, \{e\}$ of $G$, where $M_1$ and $M_2$ are perfect matchings of $G_1$ and $G_2$ containing
$e_1$ and $e_2$, respectively. Hence $G$ is matching-covered.

\medskip
\noindent\emph{(II) $G$ satisfies the vertex condition.}
Every vertex other than $x,x'$ has the same set of incident half edges from one of the $G_i$ and thus
has an incident \red{red} half-edge in $G$. The vertices
$x,x'$ inherit all incident edges (and colours) of $x_1,x_2$ and $x_1',x_2'$ and the edge removal between $x,x'$ could happen only when the vertex condition remains satisfied after the removal. Thus the vertex
condition holds in $(G,c)$.

\medskip
\noindent\emph{(III) $G$ satisfies the matching condition.}
Let $M$ be any perfect matching of $G$. As $|X|$ and $|X'|$ are odd, the number
$k := |M \cap \delta(X,X')|$ is odd. Since $E_b(G)=\delta(X,X')$, this means
$|M \cap E_b(G)|$ is odd. So it only remains to prove that
$|M \cap E_b(G)| \le 1$.

Suppose, towards a contradiction, that $M$ contains at least three bichromatic
edges. Suppose, $M$ uses the edge between $x$ and $x'$, then one of the two restrictions of $M$ to edges originating in $G_1$ and $G_2$ must have at least two bichromatic edges, a contradiction.

Without loss of generality, let $x$ be matched to a vertex in $V(G_1)\setminus \{x_1,x_1'\}$ . Observe that as $|V(G_1)\setminus \{x_1,x_1'\}|$ is even, if $x$ is matched to a vertex in $V(G_1)\setminus \{x_1,x_1'\}$ in $M$, then $x'$ is also matched to a vertex in $V(G_1)\setminus \{x_1,x_1'\}$. The restrictions of $M$ to edges originating in $G_1$ will have exactly one bichromatic edge (as $G_1$ is a W-state graph). Therefore, the restriction of $M$ to edges originating on the other side, excluding $x,x'$, say $M_2$, must have at least two bichromatic edges. But then $M_2 \cup \{e_2\}$ is a perfect matching of $G_2$ and it now contains at least two
bichromatic edges, again a contradiction. Therefore every perfect matching of
$G$ contains exactly one bichromatic edge.

Hence $(G,c)$ is a W-state graph.
\end{proof}

\section{Proof of \texorpdfstring{\Cref{mainstructthm}}{Theorem}: 3-connected components of W-state graphs} \label{sec_proof-main}

\linkedstatement{mainthm}{mainstructthm}{\mainthm*}

\begin{proof}[Proof of \Cref{mainstructthm}: ``only if'' direction]
Let $(G,c)$ be a W-state graph. By definition, it satisfies the vertex
condition, so it remains to show that each $3$-connected component of $G$ is a W-cone.

We argue by repeatedly applying the $2$-vertex cut decomposition as in
\Cref{def_W-block-decomp}. If $(G,c)$ itself is $3$-connected, then
by \Cref{lem_3conn-wcone}, $(G,c)$ is a W-cone,
and we are done.

Otherwise, $G$ has a $2$-vertex cut $\{x,x'\}$ with $x\in X$, $x'\in X'$
as in \Cref{lem_2-cut}. Let $(H_1,c_1),\dots,(H_t,c_t)$ be the W-blocks of $(G,c)$ with respect to the cut $\{x,x'\}$ as defined in
\Cref{def_W-block-decomp}. By \Cref{thm_W-blocks-are-W},
each $(H_i,c_i)$ is again a W-state graph, and each $H_i$ has strictly fewer
vertices than $G$.

We now apply the same reasoning recursively to each W-block that is not
$3$-connected. Since the number of vertices strictly decreases at each
step, this process terminates and yields a collection of induced subgraphs $(H^{(1)},c^{(1)}), \dots, (H^{(r)},c^{(r)})$ which are W-states and are $3$-connected. By \Cref{lem_3conn-wcone}, each $(H^{(j)},c^{(j)})$ is a W-cone.

Moreover, the standard decomposition of a graph along its $2$-vertex cuts implies that the graphs $H^{(1)}, \ldots, H^{(r)}$ obtained above are precisely the $3$-connected components of $G$ in the standard decomposition, with their virtual bichromatic edges. Therefore, every $3$-connected component of $G$ is a W-cone, as claimed.
\end{proof}

\begin{proof}[Proof of \Cref{mainstructthm}: ``if'' direction]
Suppose $(G,c)$ is a connected half-edge $2$-coloured, matching-covered graph such that
\begin{enumerate}[(i)]

    \item every vertex of $G$ is incident with a red half-edge, and

    \item each $3$-connected component of $G$, with the colouring inherited from $c$ on the real edges and with every virtual edge coloured bichromatically, is a W-cone.

\end{enumerate}
We prove that $(G,c)$ is a W-state graph. 

Each $3$-connected component $K$ of $G$ is a W-cone by assumption, hence a W-state graph by \Cref{wconiswstat}. Denote these components by $(K_1,d_1),\ldots,(K_r,d_r)$.

The usual $3$-connected decomposition of a $2$-connected graph expresses $G$ as a tree of such components glued along pairs of vertices (that is, along $2$-vertex cuts). By \Cref{def_W-block-decomp}, each component contains a bichromatic edge between the two vertices of the corresponding adhesion pair. Therefore each glueing step can be realized as a W-union along bichromatic edges in the sense of \Cref{def_W-union}. Starting from $K_1$ and repeatedly applying retained W-unions, and then deleting exactly those virtual bichromatic edges that were not present in $G$, we obtain a graph isomorphic to $(G,c)$. By \Cref{lem_paste-over-edge}, every intermediate graph is a W-state graph. In particular, $(G,c)$ is a W-state graph.
\end{proof}

Due to \Cref{aux_remark}, during the decomposition, one may end up creating new multi-edges. Therefore, one cannot directly rule out the possibility that a simple graph, upon decomposition, yields different multigraphs that are W-cones. The next proof rules this out.

\linkedstatement{NoSimpleWStateGraph}{simplegraphcor}{\NoSimpleWStateGraph*}
\begin{proof}
Towards a contradiction, let there be a simple W-state graph. Among all such W-state graphs, consider a vertex-minimal W-state graph $G$. As $G$ is simple, from \Cref{wcome_multi}, we know that $G$ is not a W-cone. Therefore, from \Cref{lem_2-cut}, there exists a $2$-vertex cut, say $\{x,x'\}$ with
\(x\in X\) and \(x'\in X'\).

Since $(G,c)$ is simple, there is at most one edge between $x$ and $x'$. Without loss of generality, there is no bichromatic edge between $x$ and $x'$ whose red half-edge is incident with $x$. Therefore, there exists a vertex $w\in V(G)\setminus\{x,x'\}$ such that $xw$ is a bichromatic edge whose red half-edge is incident with $x$.

Let $(H_1,c_1),\ldots,(H_t,c_t)$ be the W-blocks of $(G,c)$ obtained from the block decomposition of \Cref{def_W-block-decomp} along $\{x,x'\}$. From \Cref{thm_W-blocks-are-W}, these are W-state graphs. Without loss of generality, let $w\in V(H_1)$. By our assumption that $(G,c)$ is a vertex-minimal simple W-state graph, $(H_1,c_1)$ must be a multigraph. 

Now $x$ already has an incident red half-edge inside $H_1$, namely on $xw$. Hence, by \Cref{def_W-block-decomp}, in the construction of $H_1$ we add at most one virtual bichromatic edge between $x$ and $x'$. If the original edge $xx'$ is present in $G$, then it already provides the required bichromatic edge between $x$ and $x'$; otherwise exactly one virtual edge is added. In either case, $\{x,x'\}$ is not a multi-edge of $H_1$. Therefore, $H_1$ contains a multi-edge not between $x$ and $x'$, and that multi-edge was already present in $(G,c)$, a contradiction.
\end{proof}

\subsection{Proof of \texorpdfstring{\Cref{thm_algorithm}}
{Theorem}: W-colourability of a graph}
\label{sec_algorithms}
\smallskip 
\noindent\textbf{Faster algorithm for recognition of W-state graphs.}
We first describe a simple algorithm, implicit in \cite{vardi1,vardi2}, to verify whether a given graph is a W-state graph. One can check whether the graph is matching-covered and whether \Cref{eq_VC} (the vertex condition, referred to as FORALL-PMVC) holds in $O(|V|\cdot|E|)$ time. To verify \Cref{eq_MC} (the matching condition, referred to as EXISTS-PMVC condition), observe the following. First, one can check whether there is no perfect matching consisting entirely of monochromatic \blue{blue} edges; this certifies that every perfect matching M satisfies $|M \cap E_b(G)| \ge 1$. Second, one can check whether any pair of bichromatic edges can be simultaneously extended to a perfect matching; the absence of such a pair implies that $|M \cap E_b(G)| \le 1$ for every perfect matching M. Together, these checks establish the matching condition in $O(|E|^2)$ calls to perfect matching recognition, which, using \cite{Micaliv}, can be implemented in $O(|E|\sqrt{|V|})$ time. Therefore, 
this simple procedure requires in total $O(|E|^3 \sqrt{|V|})$ time. Our structure from \Cref{mainstructthm} gives rise to a simple algorithm which is as fast as finding a matching-covered graph, i.e., $O(|V|\cdot|E|)$ time. The recognition algorithm is given next.

\subsubsection{Recognition algorithm for W-state graphs}\label{supp_alg}

\begin{theorem}
There exists an $O(|V|\cdot |E|)$-time algorithm that, given a half-edge
$2$-coloured graph $(G,c)$, decides whether $(G,c)$ is a W-state graph.
\end{theorem}

\begin{proof}
Since every W-state graph is connected, we first test connectivity and then apply
\Cref{mainstructthm}. Thus a half-edge $2$-coloured graph $(G,c)$ is a W-state graph if and only if
\begin{enumerate}
  \item $G$ is connected and matching-covered;
  \item the vertex condition holds: every vertex has an incident \red{red} half-edge;
  \item each $3$-connected component of the standard $3$-connected decomposition of $G$,
        with the colouring inherited from $c$ on the real edges and with every virtual edge of
        the decomposition coloured bichromatically, is a W-cone.
\end{enumerate}
We describe an algorithm that checks these conditions and then analyse its running time.

\begin{algorithm}[h]
  \caption{Recognition of W-state graphs}
  \label{alg_wstate}
  \DontPrintSemicolon

  \textbf{Input:} A half-edge $2$-coloured graph $(G,c)$\;
  \textbf{Output:} \textsc{true} if and only if $(G,c)$ is a W-state graph\;

  \BlankLine
  \textbf{1. Connected + matching-covered + vertex condition.}\\
  Check whether $G$ is connected. If not, return \textsc{false}.\\
  Use Carvalho--Cheriyan~\cite{DBLP:conf/soda/CarvalhoC05} to test whether $G$ is
  matching-covered. If not, return \textsc{false}.\\
  Check the vertex condition: if some $v\in V(G)$ has no incident \red{red} half-edge,
  return \textsc{false}.

  \BlankLine
  \textbf{2. Standard $3$-connected decomposition.}\\
  Compute the standard $3$-connected decomposition of $G$ using SPQR trees~\cite{spqr2}.\\
  Let $K_1,\dots,K_t$ be the resulting components. For each component $K_i$, keep the
  colouring inherited from $c$ on the real edges and colour every virtual edge of $K_i$
  bichromatically.

  \BlankLine
  \textbf{3. Check W-cones.}\\
  For each component $K_i$, test whether $K_i$ is a W-cone using \Cref{conedef}. If some
  $K_i$ is not a W-cone, return \textsc{false}.

  \BlankLine
  \textbf{4. Return.}\\
  If all components pass the test, return \textsc{true}.
\end{algorithm}

\medskip\noindent
\textbf{Correctness.}
Suppose first that the algorithm outputs \textsc{true}. Then Step~1 ensures that $G$ is
connected and matching-covered and that the vertex condition holds in $G$. Step~2 constructs
exactly the components that appear in \Cref{mainstructthm}, namely the components of the
standard $3$-connected decomposition of $G$, with the colouring inherited from $c$ on the
real edges and every virtual edge coloured bichromatically. Step~3 checks that each such
component is a W-cone. Hence, by \Cref{mainstructthm}, $(G,c)$ is a W-state graph.

Conversely, suppose that $(G,c)$ is a W-state graph. Then $G$ is connected and
matching-covered and the vertex condition holds by \Cref{def_Wstate}, so Step~1 cannot
reject. By \Cref{mainstructthm}, every component produced in Step~2 is a W-cone. Therefore
all components pass the test in Step~3, and the algorithm outputs \textsc{true}. Thus the
algorithm returns \textsc{true} if and only if $(G,c)$ is a W-state graph.

\medskip\noindent
\textbf{Running time.}
Checking connectivity and the vertex condition in Step~1 takes $O(|V|+|E|)$ time. The
matching-covered test in Step~1 takes $O(|V|\cdot |E|)$ time by the algorithm of
Carvalho and Cheriyan~\cite{DBLP:conf/soda/CarvalhoC05}.

The SPQR-tree in Step~2 can be constructed in $O(|V|+|E|)$ time~\cite{spqr,spqr2}, and the
resulting decomposition has total size $O(|V|+|E|)$. Assigning bichromatic colours to the
virtual edges takes time linear in the size of this decomposition, so Step~2 runs in
$O(|V|+|E|)$ time overall.

For Step~3, let $n_i := |V(K_i)|$ and $m_i := |E(K_i)|$ for $i\in[t]$. Checking that $K_i$
is matching-covered takes $O(n_i m_i)$ time by Carvalho--Cheriyan~\cite{DBLP:conf/soda/CarvalhoC05},
and the additional checks for the existence of a universal vertex and for the edge-partition
and vertex conditions of \Cref{conedef} take $O(n_i+m_i)$ time, which is dominated by
$O(n_i m_i)$. Hence Step~3 runs in
\[
  \sum_{i=1}^{t} O(n_i m_i)
  \;=\;
  O\Bigl(\bigl(\max_i n_i\bigr)\cdot \sum_{i=1}^{t} m_i\Bigr)
  \;=\; O(|V|\cdot |E|),
\]
since $\max_i n_i \le |V|$ and the total size of the decomposition is linear in $|V|+|E|$,
which is $O(|E|)$ because $G$ is connected.

Combining the three steps, the overall running time of the algorithm is
$O(|V|\cdot |E|)$, as claimed.
\end{proof}

\smallskip 
\noindent\textbf{W-realizability of graphs.}
 Using \Cref{mainstructthm}, we obtain an algorithm for deciding whether an uncoloured graph $G$ is
\emph{W-realizable}, that is, whether there exists a half-edge $2$-colouring $c$ such that
$(G,c)$ is a W-state graph. At a high level, the algorithm first checks the necessary
condition that $G$ is matching-covered. It then computes the standard $3$-connected
decomposition of $G$ and checks if each resulting component is a W-cone. For every component
$K_i$, this is done by verifying that $K_i$ is matching-covered and that it contains a candidate apex,
namely a universal vertex incident with a pair of parallel edges and with every virtual edge
of $K_i$. This is exactly the structure required for $K_i$ to underlie a W-cone in a
colouring in which the virtual edges are bichromatic. Note that a $K_i$ can be realized as a W-cone in multiple ways, and the algorithm keeps track of these realizations in a way that preserves a low runtime using the SPQR-tree decomposition \cite{spqr,spqr2} and some structural properties of W-state graphs. The algorithm then checks that
these W-cone realizations can be chosen compatibly so that the original graph $G$ is
recovered by gluing the components together via W-unions along the virtual bichromatic
edges. The full algorithm is given below.

From the structural characterization in
\Cref{mainstructthm}, if $G$ is W-realizable, then in any witnessing colouring the graph is
matching-covered and each $3$-connected component of the standard decomposition, with its
virtual edges viewed as bichromatic, is a W-cone. Hence every component passes the test, and the induced colourings are compatible with the reconstruction of $G$.
Conversely, if the algorithm accepts, then each component admits a suitable W-cone
realization, and the compatibility step ensures that such realizations can be glued
together along the decomposition tree by W-unions. As W-union preserves the W-state
property, the resulting colouring of $G$ is a W-state graph. We now describe the compatibility procedure and prove the correctness.

\subsubsection{W-realizability of graphs}\label{sec:wreal}

Call an uncoloured graph $G$ \emph{W-realizable} if there exists a half-edge $2$-colouring
$c$ such that $(G,c)$ is a W-state graph in the sense of \Cref{def_Wstate}. We now prove
\Cref{thm_algorithm}.

We first isolate the condition that certifies that a 3-connected component obtained after the decomposition can be
coloured as a W-cone while forcing its virtual edges to be bichromatic.

\begin{proposition}\label{prop:forced-wcone}
Let $H$ be an uncoloured graph and let $S \subseteq E(H)$. Then there exists a half-edge
$2$-colouring $d$ such that $(H,d)$ is a W-cone and every edge of $S$ is bichromatic if and
only if $H$ is matching-covered and contains a universal vertex $v$ that is incident with a
pair of parallel edges and with every edge of $S$.
\end{proposition}

\begin{proof}
($\Rightarrow$)
Suppose that there exists a half-edge $2$-colouring $d$ such that $(H,d)$ is a W-cone and
every edge of $S$ is bichromatic. By \Cref{conedef}, $H$ is matching-covered and there exists
a universal vertex $v$ such that
\[
E_b(H)=\delta(\{v\}).
\]
Hence every edge of $S$ is incident with $v$. By \Cref{wcome_multi}, every W-cone contains a
multi-edge, and therefore $v$ is incident with a pair of parallel edges.

($\Leftarrow$)
Suppose that $H$ is matching-covered and contains a universal vertex $v$ that is incident
with a pair of parallel edges and with every edge of $S$. Let
\[
F := H[V(H)\setminus \{v\}].
\]
Choose a neighbour $u_0$ of $v$ and two distinct parallel edges $e_1,e_2$ with endpoints
$\{v,u_0\}$.

Colour every edge of $F$ monochromatically blue. Colour every edge incident with $v$
bichromatically as follows:
\[
d(e_1,v)=1,\quad d(e_1,u_0)=0,
\]
\[
d(e_2,v)=0,\quad d(e_2,u_0)=1,
\]
and for every other edge $e=\{v,u\}$ set
\[
d(e,v)=0,\qquad d(e,u)=1.
\]

Then every edge incident with $v$ is bichromatic and every edge of $F$ is monochromatically
blue. Hence
\[
\delta(\{v\})=E_b(H)
\qquad\text{and}\qquad
E(H[V(H)\setminus \{v\}])=E_m(H),
\]
so the edge-partition condition of \Cref{conedef} holds. Moreover, every vertex of
$V(H)\setminus\{v\}$ has an incident red half-edge: the vertex $u_0$ gets one from $e_2$,
and every other vertex $u$ gets one from its edge to $v$. The vertex $v$ gets an incident red
half-edge from $e_1$. Therefore the vertex condition of \Cref{conedef} also holds.

Finally, since $H$ is matching-covered and $v$ is universal, every perfect matching of $H$
contains exactly one edge incident with $v$, and all remaining matching edges lie in $F$ and
are monochromatically blue. Hence every perfect matching of $H$ contains exactly one
bichromatic edge. Therefore $(H,d)$ is a W-cone.

Since every edge of $S$ is incident with $v$, every edge of $S$ is bichromatic in the above
colouring as required.
\end{proof}

For a decomposition component $K$, let $\Sigma(K)$ be the set of distinct endpoint-pairs of
its virtual edges. We define
\[
A(K):=\left\{\, v \in V(K) \;\middle|\;
\begin{aligned}
&v \text{ is universal in } K \text{ and is incident with a pair of parallel edges,}\\
&\text{and every virtual edge of } K \text{ is incident with } v
\end{aligned}
\right\}.
\]

\begin{lemma}\label{lem:apex-candidates}
Let $K$ be a decomposition component. If $K$ admits a W-cone colouring in which every
virtual edge is bichromatic, then every feasible apex of $K$ belongs to $A(K)$. Moreover:
\begin{enumerate}
\item if $|\Sigma(K)|\geq 2$, then $A(K)$ is a singleton;
\item if $\Sigma(K)=\{\{x,y\}\}$, then $A(K)\subseteq \{x,y\}$.
\end{enumerate}
\end{lemma}

\begin{proof}
Let $d$ be a W-cone colouring of $K$ in which every virtual edge is bichromatic, and let $v$
be the apex of this W-cone. By \Cref{conedef}, all bichromatic edges of $K$ are exactly the
edges incident with $v$. Hence every virtual edge of $K$ is incident with $v$. Since $(K,d)$
is a W-cone, $v$ is universal and, by \Cref{wcome_multi}, is incident with a pair of parallel edges.
Therefore $v\in A(K)$.

Now assume that $|\Sigma(K)|\ge 2$. From \Cref{prop:forced-wcone}, as all virtual edges are incident with the apex,
every endpoint-pair in $\Sigma(K)$ must contain the apex. If there are two distinct
endpoint-pairs, then they have a unique common endpoint, and this endpoint must be the
apex. Hence $A(K)$ is a singleton.

Finally, if $\Sigma(K)=\{\{x,y\}\}$, then every virtual edge has endpoints $\{x,y\}$, and any
feasible apex must be incident with this edge. Therefore every feasible apex belongs to
$\{x,y\}$.
\end{proof}


\begin{lemma}[Compatibility along the decomposition tree]\label{compatabilitylemma}
Given the sets $A(K)$ for the nodes $K$ of the standard $3$-connected decomposition
tree $T$, the compatibility test in Step~4 can be implemented by a bottom-up dynamic
program on $T$. It accepts if and only if the decomposition components admit W-cone
colourings, with all virtual edges bichromatic, that are compatible with the standard
reconstruction of $G$; equivalently, $G$ is obtained from these coloured components by
retained W-unions followed by the legal deletion of every virtual edge not present in
$G$.
\end{lemma}

\begin{proof}
Root $T$ at an arbitrary component $R$. For a node $K$, let $T_K$ denote the subtree
rooted at $K$. If $K \neq R$, let $p(K)=\{x_K,y_K\}$ be the adhesion pair shared by
$K$ and its parent.

Fix a realization of the components in $T_K$ by retained W-unions. For
$z \in p(K)$, say that $z$ is \emph{supported from the $K$-side} if, after all retained
W-unions inside $T_K$ have been performed and every virtual edge internal to $T_K$
has been deleted, the vertex $z$ is incident with a red half-edge that does not belong
to a virtual edge of the parent adhesion pair $p(K)$.

For each non-root node $K$, we compute a set $F(K)$ of feasible states. A state is a
pair $(a,S)$, where $a \in A(K)$ and $S \subseteq p(K)$. It is feasible if there exists
a choice of W-cone colourings for the components of $T_K$ such that
\begin{enumerate}
  \item $K$ is coloured as a W-cone with apex $a$;
  \item every virtual edge in $T_K$ is bichromatic;
  \item every virtual edge internal to $T_K$ is deleted legally; and
  \item after these deletions, the vertices of $p(K)$ supported from the $K$-side are
        exactly the vertices of $S$.
\end{enumerate}
For the root, we use the same definition without the set $S$.

By \Cref{lem:apex-candidates}, $|A(K)| \le 2$ for every component $K$. Since $|p(K)| \le 2$, each table
has constant size.

Fix a component $K$ and an apex choice $a \in A(K)$. In any W-cone colouring with
apex $a$, every edge incident with $a$ is bichromatic and every other edge is
monochromatically blue. Hence, once $a$ is fixed, the only remaining freedom is the
orientation of the bichromatic edges. In particular, for each adhesion pair of $K$, the
support contributed by $K$ itself can be determined locally.

If $K$ is a leaf, we enumerate the finitely many pairs $(a,S)$ and keep exactly those
that are locally feasible.

Now assume that $K$ has children $K_1,\dots,K_t$ and that the tables $F(K_i)$ are
already known. Fix $a \in A(K)$ and process the children one by one.

If $|\Sigma(K)| \ge 2$, then $a$ is unique by \Cref{lem:apex-candidates} and every child adhesion pair
has the form $\{a,z_i\}$. Since only edges incident with the apex are bichromatic in a
W-cone, the support of $z_i$ from the $K$-side is determined locally inside $K$ once
$a$ is fixed. Thus, while processing the children, the only nonlocal information that
must be carried is which vertices of $p(K)$ are already supported and whether the apex
$a$ is already supported.

If $\Sigma(K)=\{\{x,y\}\}$, then every virtual edge of $K$ uses the same pair
$\{x,y\}$, and it suffices to remember which of $x$ and $y$ are supported.

Consider a child $K_i$ attached to $K$ along the adhesion pair $q_i=\{u_i,v_i\}$.
Combine a current partial state of $K$ with a feasible state of $K_i$. The retained
W-union along $q_i$ may be followed by the deletion of the virtual edge(s) of $q_i$
exactly when, after the union, each of $u_i$ and $v_i$ is still incident with a red
half-edge outside the deleted virtual edge(s). This is precisely the legality condition
of \Cref{def_W-union}(ii). When the condition holds, \Cref{lem_paste-over-edge} implies that the retained
W-union is a W-state graph; we then delete the virtual edge(s) of $q_i$, update the
support information, and continue. Otherwise, the pair of partial states is discarded.

Since each node has at most two apex choices and only constantly many support patterns,
each child is processed in constant time once the local support information is available.
Hence the bottom-up traversal is linear in $|T|$.

If the root has a feasible state, selecting one feasible state at each node yields
compatible W-cone colourings of all decomposition components. By construction, the
standard reconstruction of $G$ is obtained by retained W-unions, and every virtual
edge not present in $G$ is deleted only when \Cref{def_W-union}(ii) permits it. Conversely,
any compatible family of local W-cone colourings induces a feasible state at every node,
so the dynamic program accepts. This proves the lemma.
\end{proof}

\begin{algorithm}[ht]
\caption{W-colourability of a graph}
\label{alg:wstatecol}
\textbf{Input:} An uncoloured graph $G=(V,E)$ \\
\textbf{Output:} \textsc{true} if and only if there exists a half-edge $2$-colouring $c$
such that $(G,c)$ is a W-state graph
\begin{enumerate}
\item \textbf{Global matching-covered test.}
Use the algorithm of Carvalho and Cheriyan~\cite{DBLP:conf/soda/CarvalhoC05}
to test whether $G$ is a connected matching-covered graph.
If not, return \textsc{false}.

\item \textbf{3-connected decomposition.}
Compute the standard $3$-connected decomposition tree $T$ of $G$ using SPQR-trees.
Let $K_1,\dots,K_t$ be its components, each equipped with its virtual edges.

\item \textbf{Local W-cone test.}
For each $i\in [t]$:
\begin{enumerate}
\item test whether $K_i$ is matching-covered; if not, return \textsc{false};
\item compute
\[
A(K):=\left\{\, v \in V(K) \;\middle|\;
\begin{aligned}
&v \text{ is universal in } K \text{ and is incident with a pair of parallel edges,}\\
&\text{and every virtual edge of } K \text{ is incident with } v
\end{aligned}
\right\}.
\]
\item if $A_i:=A(K_i)=\emptyset$, return \textsc{false}.
\end{enumerate}

\item \textbf{Compatibility test.} Run the bottom-up dynamic program of \Cref{compatabilitylemma}.
   If the root has no feasible state, return \texttt{false}.

\item \textbf{Return.}
If all tests pass, return \textsc{true}.
\end{enumerate}
\end{algorithm}

\paragraph*{Correctness of Algorithm~\ref{alg:wstatecol}.}

Suppose first that Algorithm~\ref{alg:wstatecol} outputs true. Then Step~1 ensures that $G$ is
a connected matching-covered graph. By Step~3 and Proposition~\ref{prop:forced-wcone}, every component $K_i$ admits a
W-cone colouring in which every virtual edge is bichromatic and whose apex belongs
to $A_i$. Step~4 then checks, by \Cref{compatabilitylemma}, whether these local realizations can be
chosen compatibly so that the standard reconstruction of $G$ is obtained by retained
W-unions and every virtual edge not present in $G$ is deleted legally. By \Cref{wconiswstat}, each coloured component is a W-state
graph, and by \Cref{lem_paste-over-edge} every such W-union again yields a W-state graph. Therefore every
intermediate graph in the reconstruction is a W-state graph, and the final graph obtained at
the end of Step~4 is a half-edge $2$-coloured graph $(G,c)$ that is a W-state graph.
Hence $G$ is W-realizable.

Conversely, suppose that $G$ is W-realizable, and let $c$ be a witnessing half-edge
$2$-colouring such that $(G,c)$ is a W-state graph. By \Cref{def_Wstate}, $G$ is matching-covered,
so Step~1 does not reject. By \Cref{mainstructthm}, every $3$-connected component of the standard
decomposition of $(G,c)$, with the inherited colouring on the non-virtual edges and every
virtual edge coloured bichromatically, is a W-cone. Hence, for every component $K_i$, its
apex is universal, is incident with every virtual edge of $K_i$, and is incident with a pair of
parallel edges. Thus $A_i\neq \emptyset$ and Step~3 cannot reject. Step~4 verifies, by \Cref{compatabilitylemma}, that these local W-cone colourings can be
chosen compatibly along the decomposition tree so that the standard reconstruction of $G$ is
realised by retained W-unions and every virtual edge not present in $G$ is deleted only
when this is legal in the sense of \Cref{def_W-union}(ii). Therefore
Algorithm~\ref{alg:wstatecol} outputs \textsc{true}.

\paragraph*{Running time.}
Step~1 takes $O(|V|\cdot |E|)$ time by Carvalho and Cheriyan.
Step~2 is linear in the size of the input graph.
For Step~3, each component is tested for matching-coveredness and for the existence of a
candidate apex. Since the sum of the sizes of all decomposition components is linear in the
size of the decomposition, the total cost of Step~3 is bounded by $O(|V|\cdot |E|)$.
By \Cref{compatabilitylemma}, Step~4 is a bottom-up traversal of the decomposition tree
and therefore runs in linear time. Hence Algorithm~\ref{alg:wstatecol} runs in $O(|V|\cdot |E|)$ time.

\subsection{A complexity barrier for Dicke states}\label{dicke_section}
A natural question is whether one can extend our ideas to Dicke states, which generalise W-states by superposing all basis states of a fixed Hamming weight
$k$. In the experiment-graph language, the Dicke analogue naturally separates into two requirements.
Unlike the W-state setting, one should not exclude monochromatic red edges a priori, since
such an edge contributes two red half-edges. For a perfect matching $M$, let
\[
R(M):=\{v\in V(G): \text{the half-edge of the edge of $M$ incident with $v$ is red}\}.
\]
Then the Dicke constraints as defined in \cite{vardi1,vardi2} are:

\begin{enumerate}

\item \textsc{EXISTS-PMVC}: There exists no perfect matching $M\in PM(G)$, such that
$|R(M)|\neq k$.
\item \textsc{FORALL-PMVC}: For all $S\subseteq V(G)$ with $|S|=k$, there exists
a perfect matching $M\in PM(G)$ such that $R(M)=S$.
\end{enumerate}

To isolate this barrier, we use the general PMVC framework of \cite{vardi1,vardi2}. Given a vertex colouring
$\varphi:V(G)\to\{0,1\}$, let $G_\varphi$ be the simple graph on $V(G)$ where
$uv\in E(G_\varphi)$ if and only if $G$ contains an edge $e=uv$ with
$c(e,u)=\varphi(u)$ and $c(e,v)=\varphi(v)$. Equivalently, $\varphi$ is induced by some
perfect matching of $(G,c)$ if and only if $G_\varphi$ has a perfect matching. For
$k\in\{0,\ldots,|V(G)|\}$, let
\[
\mathcal{C}_k(G):=\{\varphi:V(G)\to\{0,1\}\mid |\varphi^{-1}(1)|=k\}.
\]

\begin{definition}[Dicke \textsc{FORALL-PMVC}]
Given a half-edge $2$-coloured multigraph $(G,c)$ and an integer $k$, decide whether
$G_\varphi$ has a perfect matching for every $\varphi\in\mathcal{C}_k(G)$.
\end{definition}

\begin{restatable}[\cite{vardi1}]{observation}{ExistsPMVCPolynomial}
\label{obs:exists-pmvc-poly}
EXISTS-PMVC can be checked in
polynomial time.
\end{restatable}
\begin{proof}
Assign weight $2$ to each monochromatic red edge, weight $1$ to each
bichromatic edge, and weight $0$ to each monochromatic blue edge. Then the weight of a
perfect matching $M$ is exactly $|R(M)|$. Hence EXISTS-PMVC holds if and only if the
minimum-weight and maximum-weight perfect matchings both have weight $k$.
\end{proof}

Thus the  difficulty lies in verifying FORALL-PMVC to identify Dicke-state graphs.
\begin{restatable}[\cite{vardi2}]{observation}{DickeForallPMVCInCoNP}
\label{conp_obs}
Dicke \textsc{FORALL-PMVC} is in coNP.
\end{restatable}
\begin{proof}
A NO-certificate is a colouring
$\varphi\in\mathcal{C}_k(G)$ such that $G_\varphi$ has no perfect matching. Given
$\varphi$, the graph $G_\varphi$ can be constructed in polynomial time, and one can test
whether it has a perfect matching in polynomial time.
\end{proof}

The hardness of Dicke \textsc{FORALL-PMVC} was not known \cite{vardi2}. We resolve this.

\linkedstatement{DickeForallPMVCCoNPComplete}{thm:dicke-forall-pmvc-conp-complete}{\DickeForallPMVCCoNPComplete*}
\begin{proof}
Membership in coNP follows from \cref{conp_obs}. For coNP-hardness, it is enough to show that the complement of the problem is NP-hard.
We reduce from \textsc{Vertex Cover}. Let $(H=(V_H,E_H),k)$ be an instance of
\textsc{Vertex Cover}. We construct a half-edge $2$-coloured multigraph $(G,c)$ as follows.

For every $v\in V_H$, create two vertices $v^0$ and $v^1$, and add two further vertices
$a$ and $b$. Thus
\[
V(G)=\{v^0,v^1\mid v\in V_H\}\cup\{a,b\}.
\]
Now add the following edges.
\begin{enumerate}
\item For each $v\in V_H$, add four parallel edges between $v^0$ and $v^1$, one of each
type $(0,0)$, $(0,1)$, $(1,0)$, and $(1,1)$. Hence $v^0v^1\in E(G_\varphi)$ for every
vertex colouring $\varphi$.

\item For each $\{u,v\}\in E_H$, add one edge between $u^0$ and $v^0$ of type $(0,0)$.
Hence $u^0v^0\in E(G_\varphi)$ if and only if $\varphi(u^0)=\varphi(v^0)=0$.

\item For each $v\in V_H$, add four parallel edges between $a$ and $v^1$, and four parallel
edges between $b$ and $v^1$, again one of each type $(0,0)$, $(0,1)$, $(1,0)$, and
$(1,1)$. Thus $av^1,bv^1\in E(G_\varphi)$ for every vertex colouring $\varphi$.

\item Add no edge between $a$ and $b$, and no other edges.
\end{enumerate}

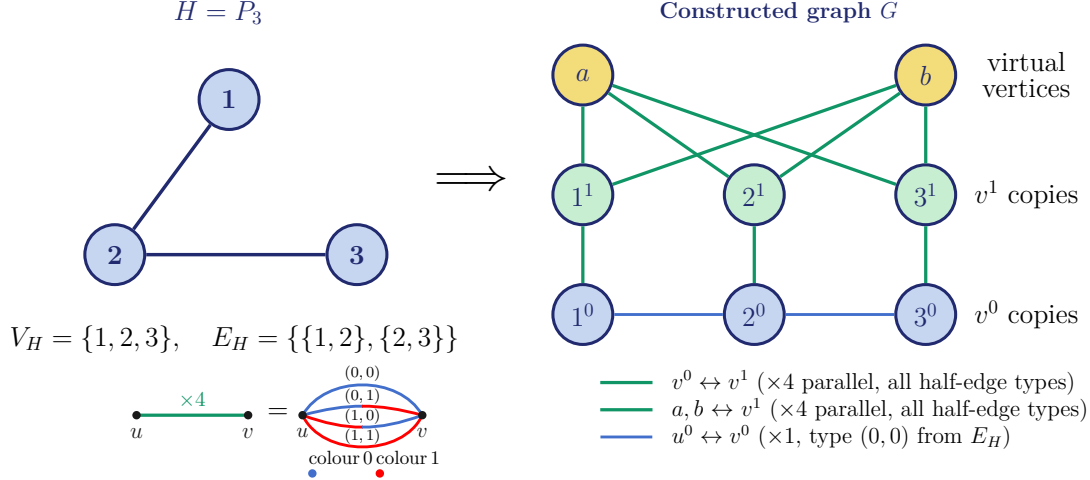
\begin{figure}
\scalebox{.72}{
 \begin{tikzpicture}[x=1cm,y=1cm, font=\sffamily]


\node[title] at (3.7,6.7) {\Large $H = P_3$};

\node[leftvertex] (L2) at (1.8,2.25) {\Large 2};
\node[leftvertex] (L1) at (3.9,5.1) {\Large 1};
\node[leftvertex] (L3) at (6.25,2.25) {\Large 3};
\draw[darkedge] (L2) -- (L1);
\draw[darkedge] (L2) -- (L3);

\node at (4.0,0.7)  {{\Large $V_H = \{1,2,3\},\quad E_H = \{\{1,2\},\{2,3\}\}$}};

\node[scale=2.2] at (9.35-1,3.6) {$\Longrightarrow$};

\node[title] at (16.0-2,6.7) {\large Constructed graph $G$};

\node[auxvertex] (a)   at (12.4-2,5.55) {$a$};
\node[auxvertex] (b)   at (18.7-2,5.55) {$b$};
\node[vone]      (v11) at (12.4-2,3.35) {$1^{1}$};
\node[vone]      (v21) at (15.55-2,3.35) {$2^{1}$};
\node[vone]      (v31) at (18.7-2,3.35) {$3^{1}$};
\node[vzero]     (v10) at (12.4-2,1.15) {$1^{0}$};
\node[vzero]     (v20) at (15.55-2,1.15) {$2^{0}$};
\node[vzero]     (v30) at (18.7-2,1.15) {$3^{0}$};

\draw[gedge] (a) -- (v11);
\draw[gedge] (a) -- (v21);
\draw[gedge] (a) -- (v31);
\draw[gedge] (b) -- (v11);
\draw[gedge] (b) -- (v21);
\draw[gedge] (b) -- (v31);
\draw[gedge] (v11) -- (v10);
\draw[gedge] (v21) -- (v20);
\draw[gedge] (v31) -- (v30);

\draw[bedge] (v10) -- (v20) -- (v30);

\node[legendtext] at (21.35-2.8,5.55+.25) {\Large virtual};
\node[legendtext] at (21.35-2.8,5.55-.5+.25) {\Large vertices};
\node[legendtext] at (21.35-2.8,3.35) {\Large $v^{1}$ copies};
\node[legendtext] at (21.35-2.8,1.15) {\Large $v^{0}$ copies};

\coordinate (lu) at (2.2,-0.2-.5);
\coordinate (lv) at (4.25,-0.2-.5);
\node[font=\bfseries\small, text=edgegreen] at ($(lu)!0.5!(lv)+(0,0.18+.1)$) {\normalsize $\times 4$};
\node[font=\Large] at (5.1-.3,-0.15-.5) {$=$};
\draw[gedge] (lu) -- (lv);
\fill[black!90] (lu) circle (2.2pt) node[below=3pt, font=\bfseries\large] {$u$};
\fill[black!90] (lv) circle (2.2pt) node[below=3pt, font=\bfseries\large] {$v$};

\begin{scope}[shift={(5.25,-0.2)}]
  \coordinate (u) at (0,0-.5);
  \coordinate (v) at (2.2,0-.5);
  \coordinate (mid) at ($(u)!0.5!(v)$);

  \draw[bedge] (u) to[bend left=60] (v);
  \node at ($(mid)+(0,0.72+.06)$) {\scriptsize $(0,0)$};

  \begin{scope}
    \clip (-0.3,-1.0) rectangle ($(mid)+(0,1.15)$);
    \draw[bedge] (u) to[bend left=15] (v);
  \end{scope}
  \begin{scope}
    \clip ($(mid)+(0,-1.0)$) rectangle (2.5,1.15);
    \draw[draw=red, line width=1.5pt] (u) to[bend left=15] (v);
  \end{scope}
  \node at ($(mid)+(0,0.35)$) {\scriptsize $(0,1)$};

  \begin{scope}
    \clip (-0.3,-1.0) rectangle ($(mid)+(0,1.0)$);
    \draw[draw=red, line width=1.5pt] (u) to[bend right=20] (v);
  \end{scope}
  \begin{scope}
    \clip ($(mid)+(0,-1.0)$) rectangle (2.5,1.0);
    \draw[bedge] (u) to[bend right=20] (v);
  \end{scope}
  \node at ($(mid)+(0,-0.38+.36)$) {\scriptsize $(1,0)$};

  \draw[draw=red, line width=1.5pt] (u) to[bend right=65] (v);
  \node at ($(mid)+(0,-0.78+.4)$) {\scriptsize $(1,1)$};

  \fill[edgeblue] (0.18,-1.07-.5) circle (2.2pt);
  \node[anchor=west, font=\small] at (0.28-.3,-1.07-.2) {colour\! 0};
  \fill[red] (1.42,-1.07-.5) circle (2.2pt);
  \node[anchor=west, font=\small] at (1.52-.3,-1.07-.2) {colour\! 1};

   \fill[black!90] (u) circle (2.2pt) node[below=3pt] {$u$};
  \fill[black!90] (v) circle (2.2pt) node[below=3pt] {$v$};
  
\end{scope}

\draw[gedge, line cap=rect] (10.75,-0.1) -- (11.65,-0.1);
\node[legendtext, anchor=west] at (11.9,-0.1) {$v^{0}\leftrightarrow v^{1}$ ($\times 4$ parallel, all half-edge types)};

\draw[gedge, line cap=rect] (10.75,-0.6) -- (11.65,-0.6);
\node[legendtext, anchor=west] at (11.9,-0.6) {$a,b$ $\leftrightarrow$ $v^{1}$ ($\times 4$ parallel, all half-edge types)};

\draw[bedge, line cap=rect] (10.75,-1.1) -- (11.65,-1.1);
\node[legendtext, anchor=west] at (11.9,-1.1) {$u^{0}\leftrightarrow v^{0}$ ($\times 1$, type $(0,0)$ from $E_H$)};
\end{tikzpicture}
}
\caption{Construction illustrated on $H = P_3$. 
Vertices $v^0, v^1$ are the two copies of each $v \in V_H$; vertices $a, b$ are virtual.}
 \label{fig:conp-construction}
\end{figure}


Fix a vertex colouring $\varphi:V(G)\to\{0,1\}$ and define
\[
C(\varphi):=\{v\in V_H\mid \varphi(v^0)=1\}.
\]
We claim that
\[
G_\varphi\text{ has a perfect matching}
\iff
C(\varphi)\text{ is not a vertex cover of }H.
\]

Suppose first that $C(\varphi)$ is not a vertex cover of $H$. Then there exists an edge
$\{u,v\}\in E_H$ with $u,v\notin C(\varphi)$. Hence $\varphi(u^0)=\varphi(v^0)=0$, so the
edge $u^0v^0$ belongs to $G_\varphi$. Now
\[
\{u^0v^0,au^1,bv^1\}\cup\{w^0w^1\mid w\in V_H\setminus\{u,v\}\}
\]
is a perfect matching of $G_\varphi$.

Conversely, suppose that $C(\varphi)$ is a vertex cover of $H$. Then every edge of $H$ has
at least one endpoint in $C(\varphi)$, so no edge between two $0$-copies survives in
$G_\varphi$. Therefore each vertex $v^0$ can only be matched to $v^1$, and any perfect
matching of $G_\varphi$ would have to contain all edges $v^0v^1$. But then every vertex
$v^1$ is already matched, while the vertices $a$ and $b$ have neighbours only among the
$v^1$-vertices and are not adjacent to each other. Hence $G_\varphi$ has no perfect
matching. This proves the claim.

We now complete the reduction. If $H$ has a vertex cover $C\subseteq V_H$ with $|C|\leq k$,
define a colouring $\varphi\in\mathcal{C}_k(G)$ by setting $\varphi(v^0)=1$ for
$v\in C$, $\varphi(v^0)=0$ for $v\notin C$, and placing the remaining $k-|C|$ red vertices
arbitrarily among $\{v^1\mid v\in V_H\}\cup\{a,b\}$. Then $C(\varphi)=C$ is a vertex cover,
so by the claim $G_\varphi$ has no perfect matching. Thus the constructed instance is a
NO-instance of Dicke \textsc{FORALL-PMVC}.

Conversely, suppose the constructed instance is a NO-instance of
Dicke \textsc{FORALL-PMVC}. Then there exists a colouring
$\varphi\in\mathcal{C}_k(G)$ such that $G_\varphi$ has no perfect matching. By the claim,
$C(\varphi)$ is a vertex cover of $H$, and clearly $|C(\varphi)|\leq k$ because
$C(\varphi)$ counts only some of the $k$ red vertices of $\varphi$. Hence $(H,k)$ is a
YES-instance of \textsc{Vertex Cover}.

Therefore $(H,k)$ is a YES-instance of \textsc{Vertex Cover} if and only if the constructed
instance is a NO-instance of Dicke \textsc{FORALL-PMVC}. The complement of
Dicke \textsc{FORALL-PMVC} is NP-hard, so Dicke \textsc{FORALL-PMVC} is coNP-hard.
Together with membership in coNP, this proves the theorem.
\end{proof}
\section{Conclusion and open problems}
We gave a complete structural characterization of W-state graphs, a graph-theoretic abstraction of experiments generating multipartite W-states. As consequences, we showed that no W-state graph is a simple graph and obtained an $O(|V|\cdot|E|)$-time recognition algorithm. Conceptually, our results place W-state graphs firmly within classical matching theory: the monochromatic edges form two factor-critical components, the cut between them is tight, and the 3-connected components are exactly W-cones, obtained from factor-critical bases by adding a universal vertex. We end with a few natural open problems.

In sharp contrast with the W-state case, the Dicke state analogue encounters a complexity barrier at the realizability level: EXISTS-PMVC is easy to verify, but FORALL-PMVC is coNP-complete. Understanding the hardness of FORALL-PMVC when restricted to graphs satisfying the EXISTS-PMVC condition for Dicke states is an interesting open question. Obtaining a structural characterization analogous to \Cref{mainstructthm} under these constraints is another natural open question.

The experiment-graph framework has also been developed to capture more sophisticated situations involving \textsl{destructive interference} \cite{Quantum_graphs_2, Quantum_graphs_3}. In that setting, one allows edge weights to be complex numbers and requires that they satisfy an appropriate system of non-linear equations. The unweighted case considered here corresponds to the regime in which all edge weights are positive real numbers, i.e., to circuits with no phase-shift gates \cite{Quantum_graphs}. It would be interesting to obtain analogous structural characterizations for versions of W-state graphs
where edges are weighted by complex numbers. More broadly, identifying quantum states that cannot be created by constructive interference alone but can only be realized using destructive interference remains an important open problem. 

W-state graphs also realize a “for all” version of the Exact Matching problem with parameter $k = 1$: every perfect matching contains exactly one \red{red} edge. Given a \red{red}–\blue{blue} edge-coloured graph $H$, the property that every perfect matching of $H$ has exactly $k$ \red{red} edges can be checked in polynomial time: assign weight $0$ to \blue{blue} edges and $1$ to \red{red} edges, and compute a minimum-weight and a maximum-weight perfect matching; the property holds if and only if both have total weight $k$. However, several related questions remain poorly understood. Given an uncoloured graph $G$ and an integer $k$, can we $2$-colour its edges so that every perfect matching of $G$ has \textsl{exactly} $k$ \red{red} edges? For which graphs and parameters does such a colouring exist? A better understanding of such universal Exact Matching constraints could give us insight into new ways of designing experiments that create more general quantum states and help precisely determine which experiments are possible.

\paragraph*{Acknowledgements.}
We thank Nishad Kothari for simplifying the proof by pointing us to the result of Edmonds, Lovász, and Pulleyblank (\Cref{thm_brickchar}).

\newpage 

\bibliography{literature_w_state}

\end{document}